\documentclass[11pt,a4paper]{article}

\usepackage{amsmath,amsfonts,amssymb,amsthm}
\usepackage{amsthm}
\usepackage{graphicx,color}
\usepackage{boxedminipage}
\usepackage{framed}
\usepackage{thmtools}
\usepackage{thm-restate}
\usepackage{xspace}
\usepackage{geometry}
 \usepackage[ pdftex, plainpages = false, pdfpagelabels, 
                 bookmarks=false,
                 bookmarksopen = true,
                 bookmarksnumbered = true,
                 breaklinks = true,
                 linktocpage,
                 pagebackref,
                 colorlinks = true,  
                 linkcolor = blue,
                 urlcolor  = blue,
                 citecolor = red,
                 anchorcolor = green,
                 hyperindex = true,
                 hyperfigures
                 ]{hyperref} 
 \usepackage{xifthen}
 \usepackage{tabularx}
\usepackage{tikz} 
 \usetikzlibrary{calc}

\DeclareMathOperator{\operatorClassP}{P}
\newcommand{\classP}{\ensuremath{\operatorClassP}}
\DeclareMathOperator{\operatorClassNP}{NP}
\newcommand{\classNP}{\ensuremath{\operatorClassNP}}
\DeclareMathOperator{\operatorClassCoNP}{coNP}
\newcommand{\classCoNP}{\ensuremath{\operatorClassCoNP}}
\DeclareMathOperator{\operatorClassFPT}{FPT\xspace}
\newcommand{\classFPT}{\ensuremath{\operatorClassFPT}\xspace}
\DeclareMathOperator{\operatorClassW}{W}
\newcommand{\classW}[1]{\ensuremath{\operatorClassW[#1]}}

\newcommand{\cO}{\mathcal{O}}
\newcommand{\Oh}{\mathcal{O}}

\newcommand{\bran}[1]{branchable\xspace}

\newtheorem{theorem}{Theorem}
\newtheorem{lemma}{Lemma}

\newtheorem{corollary}{Corollary}

\newtheorem{proposition}{Proposition}

\theoremstyle{definition}
\newtheorem{reduction}{Reduction Rule}[section]

\newcommand{\pname}{\textsc}
\newcommand{\ProblemFormat}[1]{\pname{#1}}
\newcommand{\ProblemIndex}[1]{\index{problem!\ProblemFormat{#1}}}
\newcommand{\ProblemName}[1]{\ProblemFormat{#1}\ProblemIndex{#1}{}\xspace}

\newcommand{\probLSS}{\ProblemName{Secluded $\Pi$-Subgraph}}
\newcommand{\probCSS}{\ProblemName{Connected Secluded $\Pi$-Subgraph}}
\newcommand{\probLCSS}{\ProblemName{Connected Secluded $\Pi$-Subgraph}}
\newcommand{\probCSWCS}{\ProblemName{Connected Secluded  Colored $\Pi$-Subgraph of Exact Size}}
\newcommand{\probLForb}{\ProblemName{Connected Secluded $\mathcal{F}$-Free Subgraph}}
\newcommand{\probMAXF}{\ProblemName{Maximum or $w$-Weighted Connected Secluded $\mathcal{F}$-Free Subgraph}}
\newcommand{\probBORDF}{\ProblemName{Bordered Maximum or $w$-Weghted Connected Secluded $\mathcal{F}$-Free Subgraph}}

\newcommand{\probLT}{\ProblemName{Large Secluded Tree}}
\newcommand{\probMAXT}{\ProblemName{Maximum or $w$-Weighted Secluded Tree}}
\newcommand{\probBORDT}{\ProblemName{Bordered Maximum or $w$-Weghted  Secluded Tree}}

\newcommand{\probSC}{\ProblemName{Secluded Clique}}
\newcommand{\probSSt}{\ProblemName{Secluded Star}}
\newcommand{\probSP}{\ProblemName{Secluded Long Path}}
\newcommand{\probCSDRS}{\ProblemName{Connected Secluded Regular Subgraph}}



\newlength{\RoundedBoxWidth}
\newsavebox{\GrayRoundedBox}
\newenvironment{GrayBox}[1]%
   {\setlength{\RoundedBoxWidth}{.93\textwidth}
    \def\boxheading{#1}
    \begin{lrbox}{\GrayRoundedBox}
       \begin{minipage}{\RoundedBoxWidth}}%
   {   \end{minipage}
    \end{lrbox}
    \begin{center}
    \begin{tikzpicture}%
       \node(Text)[draw=black!20,fill=white,rounded corners,%
             inner sep=2ex,text width=\RoundedBoxWidth]%
             {\usebox{\GrayRoundedBox}};
        \coordinate(x) at (current bounding box.north west);
        \node [draw=white,rectangle,inner sep=3pt,anchor=north west,fill=white] 
        at ($(x)+(6pt,.75em)$) {\boxheading};
    \end{tikzpicture}
    \end{center}}     

\newenvironment{defproblemx}[2][]{\noindent\ignorespaces%
                                \FrameSep=6pt%
                                \parindent=0pt%
                \vspace*{-1.5em}
                \ifthenelse{\isempty{#1}}{%
                  \begin{GrayBox}{\textsc{#2}}%
                }{%
                  \begin{GrayBox}{\textsc{#2} parameterized by~{#1}}%
                }
                \begin{tabular*}{\textwidth}{@{\hspace{.1em}} >{\itshape} p{1.8cm} p{0.8\textwidth} @{}}%
            }{
                \end{tabular*}%
                \end{GrayBox}%
                \ignorespacesafterend
            }

\newcommand{\defproblema}[3]{
  \begin{defproblemx}{#1}
    Input:  & #2 \\
    Task: & #3
  \end{defproblemx}
}%

\begin{document}

\title{Finding Connected Secluded Subgraphs\thanks{The preliminary version of this paper appeared as an extended abstract in the proceedings of IPEC 2017. This work is supported by Research Council of Norway via project ``CLASSIS''.}
}

\author{
Petr A. Golovach\thanks{
Department of Informatics, University of Bergen, Norway, \texttt{\{petr.golovach, pinar.heggernes, paloma.lima\}@ii.uib.no}.} \addtocounter{footnote}{-1}
\and
Pinar Heggernes\footnotemark{}\addtocounter{footnote}{-1}
\and 
Paloma Lima\footnotemark{}
\and
Pedro Montealegre\thanks{Facultad de Ingenier\'{\i}a y Ciencias, Universidad Adolfo Ib\'a\~nez,
Santiago, Chile, \texttt{p.montealegre@uai.cl}.}
}

\date{}

\maketitle

\begin{abstract}
Problems related to finding induced subgraphs satisfying given properties form one of the most studied areas within graph algorithms. Such problems have given rise to breakthrough results and led to development of new techniques both within the traditional \classP \/ vs \classNP \/ dichotomy and within parameterized complexity.  The $\Pi$-{\sc Subgraph} problem asks whether an input graph contains an induced subgraph on at least $k$ vertices satisfying graph property $\Pi$. For many applications, it is desirable that the found subgraph has as few connections to the rest of the graph as possible, which gives rise to the {\sc Secluded} $\Pi$-{\sc Subgraph} problem. Here, input $k$ is the size of the desired subgraph, and input $t$ is a limit on the number of neighbors this subgraph has in the rest of the graph. This problem has been studied from a parameterized perspective, and unfortunately it turns out to be \classW{1}-hard for many graph properties $\Pi$, even when parameterized by $k+t$. We show that the situation changes when we are looking for a connected induced subgraph satisfying $\Pi$.  In particular, we show that the \probLCSS problem is \classFPT
when parameterized by just $t$ for many important graph properties $\Pi$. 
\end{abstract}

\section{Introduction}\label{sec:intro}

Vertex deletion problems are central in parameterized algorithms and complexity, and they have contributed hugely to the development of new algorithmic methods. The $\Pi$-{\sc Deletion} problem, with input  a graph $G$ and an integer $\ell$, asks whether at most $\ell$ vertices can be deleted from $G$ so that the resulting graph satisfies graph property $\Pi$. Its dual, the $\Pi$-{\sc Subgraph} problem, with input $G$ and $k$, asks whether $G$ contains an induced subgraph on at least $k$ vertices satisfying $\Pi$. The problems were introduced already in 1980 by Yannakakis and Lewis \cite{LewisY80}, who showed their \classNP-completeness for almost all interesting graph properties $\Pi$. During the last couple of decades, these problems have been studied extensively with respect to parameterized complexity and kernelization, which has resulted in numerous new techniques and methods in these fields \cite{CyganFKLMPPS15, DowneyF13}. 

In many network problems, the size of the {\it boundary} between the subgraph that we are looking for and the rest of the graph makes a difference.  A small boundary limits the exposure of the found subgraph, and notions like isolated cliques have been studied in this respect \cite{HuffnerKMN09, ItoI09, KomusiewiczHMN09}. Several measures for the boundary have been proposed; in this work we use the open neighborhood of the returned induced subgraph. For a set of vertices $U$ of a graph $G$ and a positive integer $t$, we say that $U$ is \emph{$t$-secluded} if $|N_G(U)|\leq t$. Analogously, 
an induced subgraph $H$ of $G$ is \emph{$t$-secluded} if the vertex set of $H$ is $t$-secluded. For a given graph property $\Pi$, we get the following formal definition of the problem \probLSS. 

\defproblema{\probLSS}%
{A graph $G$ and nonnegative integers $k$ and $t$.}%
{Decide whether $G$ contains a $t$-secluded induced subgraph $H$ on at least $k$ vertices, satisfying $\Pi$.}

Lewis and Yannakakis \cite{LewisY80} showed that $\Pi$-{\sc Subgraph} is \classNP-complete for every hereditary nontrivial graph property $\Pi$. This immediately implies that \probLSS is \classNP-complete for every such $\Pi$.  As a consequence, the interest has shifted towards the parameterized  complexity of the problem, which has been studied by van Bevern et al.~\cite{BevernFMMSS16} for several classes $\Pi$. Unfortunately, in most cases   \probLSS  proves to be \classW{1}-hard, even when parameterized by $k+t$. 
In particular, it is \classW{1}-hard to decide whether a graph $G$ has a $t$-secluded independent set of size $k$ when the problem is parameterized by $k+t$ \cite{BevernFMMSS16}. 
In this extended abstract, we 
show that the situation changes when the secluded subgraph we are looking 
for is required to be connected, in which case we are able to obtain positive results that apply to many properties $\Pi$. In fact, connectivity is central in recently studied variants of secluded subgraphs, like {\sc Secluded Path} \cite{ChechikJPP13,JohnsonLR14} and {\sc Secluded Steiner Tree} \cite{FominGKK16}. 
However, in these problems the boundary measure is the closed neighborhood of the desired path or the steiner tree, connecting a given set of vertices. 
The following formal definition describes the problem that we study in this extended abstract, \probLCSS. For generality, we define a weighted problem.  

 \defproblema{\probLCSS}%
{A graph $G$, a weight function $\omega\colon V(G)\rightarrow \mathbb{Z}_{>0}$, a nonnegative integer $t$ and a positive integer $w$.}%
{Decide whether $G$ contains a connected $t$-secluded induced subgraph $H$ with $\omega(V(H))\geq w$,  satisfying $\Pi$.}

Observe that \probLCSS \/ remains \classNP-complete for all hereditary nontrivial graph properties $\Pi$, following the 
results of 
Yannakakis \cite{Yannakakis79}.
It can be also seen that \probLCSS parameterized by $w$ is \classW{1}-hard even for unit weights, if it is \classW{1}-hard with parameter $k$ to decide whether $G$ has a connected induced subgraph on at least $k$ vertices, satisfying $\Pi$ (see, e.g., \cite{DowneyF13, PapadimitriouY96}).

It is thus more interesting to consider parameterization by $t$, and we show that \probLCSS   is fixed parameter tractable when parameterized by $t$ for many important graph properties $\Pi$. 
Our main result is given in Section~\ref{sec:fpt-forb}  where we consider \probLCSS  for all graph properties $\Pi$ that are characterized by  finite sets $\mathcal{F}$ of forbidden induced subgraphs and refer to this variant of the problem as \probLForb. We show that the problem is fixed parameter tractable when parameterized by $t$ by proving the following theorem.

\begin{theorem}\label{thm:forb}
\probLForb can be solved in time  $2^{2^{2^{\Oh(t\log t)}}}\cdot n^{\Oh(1)}$.
\end{theorem}

In Section~\ref{sec:fpt-tree} we show that similar results could be obtained for some infinite families of forbidden subgraphs by giving an \classFPT algorithm for the case when  $\Pi$ is the property of being acyclic. Furthermore, in Section~\ref{sec:better} we show that we can get faster algorithms for \probLCSS when $\Pi$ is the property of being a complete graph, a star, a path, and a $d$-regular graph. Finally, in Section~\ref{sec:concl} we briefly discuss kernelization for \probLCSS.

\section{Preliminaries}\label{sec:defs}
We consider only finite undirected simple graphs. We use $n$ to denote the number of vertices and $m$ the number of edges of the considered graphs unless it creates confusion. A graph $G$ is identified by its vertex set $V(G)$ and edge set $E(G)$. For $U\subseteq V(G)$, we write $G[U]$ to denote the subgraph of $G$ induced by $U$. 
We write $G-U$ to denote the graph $G[V(G)\setminus U]$; for a single-element $U=\{u\}$, we write $G-u$.
A set of vertices $U$ is \emph{connected} if $G[U]$ is a connected graph.
For a vertex $v$, we denote by $N_G(v)$ the \emph{(open) neighborhood} of $v$ in  $G$, i.e., the set of vertices that are adjacent to $v$ in $G$. For a set $U\subseteq V(G)$, $N_G(U)=(\cup_{v\in U}N_G(v))\setminus U$. 
We denote by $N_G[v]=N_G(v)\cup\{v\}$ the \emph{closed neighborhood} of $v$; respectively, $N_G[U]=\cup_{v\in U}N_G[v]$.
The \emph{degree} of a vertex $v$ is $d_G(v)=|N_G(v)|$. 
Two vertices $u$ and $v$ of graph $G$ are {\it true twins} if $N_G[u]=N_G[v]$, and {\it false twins} if $N_G(u)=N_G(v)$.
A set of vertices $S\subset V(G)$ of a connected graph $G$ is a \emph{separator} if $G-S$ is disconnected. 
A vertex $v$ is a \emph{cut vertex} if $\{v\}$ is a separator.

A graph property is {\it hereditary} if it is preserved under vertex deletion, or equivalently, under taking induced subgraphs. A graph property is {\it trivial} if either the set of graphs satisfying it, or the set of graphs that do not satisfy it, is finite. 
Let $F$ be a graph. 
We say that a graph $G$ is \emph{$F$-free} if $G$ has no induced subgraph isomorphic to $F$. For a set of graphs $\mathcal{F}$, a graph $G$ is \emph{$\mathcal{F}$-free} if $G$ is $F$-free for every $F\in \mathcal{F}$. Let $\Pi$ be the property of being $\mathcal{F}$-free. Then, depending on whether $\mathcal{F}$ is a finite or an infinite set, we say that $\Pi$ is {\it characterized by a finite / infinite set of forbidden induced subgraphs}.

We use the \emph{recursive understanding} technique introduced by Chitnis et al.~\cite{ChitnisCHPP16} for graph problems to solve \probLCSS when $\Pi$ is defined by forbidden induced subgraphs or $\Pi$ is the property to be a forest. This powerful technique is based on the following idea.  Suppose that the input graph has a vertex separator of bounded size that separates the graph into two sufficiently big parts. Then we solve the problem recursively for one of the parts and replace this part by an equivalent graph such that the replacement keeps all essential (partial) solutions of the original part. By such a replacement we obtain a graph of smaller size. Otherwise, if there is no separator of bounded size separating graphs into two big parts, then either the graph has bounded size or it is highly connected, and we exploit these properties. We need the following notions and results from Chitnis et al.~\cite{ChitnisCHPP16}.

Let $G$ be a graph. A pair $(A,B)$, where $A,B\subseteq V(G)$ and $A\cup B=V(G)$, is a \emph{separation of $G$ of order $|A\cap B|$} if $G$ has no edge $uv$ with $u\in A\setminus B$ and $v\in B\setminus A$, i.e., $A\cap B$ is an $(A,B)$-separator. Let $q$ and $k$ be nonnegative integers. A graph $G$  is \emph{(q,k)-unbreakable} if for every separation $(A,B)$ of $G$ of order at most $k$, $|A\setminus B|\leq q$ or $|B\setminus A|\leq q$. Combining Lemmas~19, 20 and 21 of~\cite{ChitnisCHPP16}, we obtain the following.

\begin{lemma}[\cite{ChitnisCHPP16}]\label{lem:unbreakable}
Let $q$ and $k$ be nonnegative integers. There is an algorithm with running time $2^{\Oh(\min\{q,k\}\log (q+k))}\cdot n^3\log n$ that, for a 
graph $G$, either finds a separation $(A,B)$ of order at most $k$ such that $|A\setminus B|>q$ and $|B\setminus A|>q$, or correctly reports that $G$ is $((2q+1)q\cdot 2^k,k)$-unbreakable. 
\end{lemma}

We conclude this section by noting that the following variant of  \probCSS is \classFPT when parameterized by $k+t$.
We will rely on this result in the subsequent sections, however we believe that it is also of interest on its own.

\defproblema{\probCSWCS}%
{A graph $G$, coloring $c\colon V(G)\rightarrow \mathbb{N}$, a weight function $\omega\colon V(G)\rightarrow\mathbb{Z}_{\geq 0}$ and nonnegative integers $k$, $t$ and $w$.}%
{Decide whether $G$ contains a connected $t$-secluded induced subgraph $H$ such that $(H,c')$, where $c'(v)=c|_{V(H)}(v)$, satisfies  $\Pi$, $|V(H)|=k$ and 
$\omega(V(H))\geq w$.}

We say that a mapping $c\colon V(G)\rightarrow \mathbb{N}$ is a \emph{coloring} of $G$; note that we do not demand a coloring to be proper. Analogously, we say that $\Pi$ is a \emph{property of colored graphs} if $\Pi$ is a property on pairs $(G,c)$, where $G$ is a graph and $c$ is a coloring.  
Notice that if some vertices of the input graph have labels, then we can assign to each label (or a combination of labels if a vertex can have several labels) a specific color and assign some color to unlabeled vertices. Then we can redefine a considered graph property with the conditions imposed by labels as a property of colored graphs. 
Observe that we allow zero weights. 
We give two algorithms for \probCSWCS with different running times. The first algorithm is based on Lemmas~3.1 and 3.2 of Fomin and Villanger~\cite{FominV12}, which we summarize in Lemma \ref{lem:conn-subgr} below. The second algorithm uses Lemma~\ref{lem:derand} by Chitnis et al.~\cite{ChitnisCHPP16},  and we are going to use it when $k\gg t$. 

\begin{lemma}[\cite{FominV12}]\label{lem:conn-subgr} Let $G$ be a graph. For every $v\in
V(G)$, and $k,t\geq 0$,  the number of connected vertex subsets
$U\subseteq V(G)$ such that
$v \in U$,
$|U| = k$, and 
$|N_G(U)|=t$,  
is at most $\binom{k+t}{t}$.
Moreover, all such subsets can be enumerated in time $\cO (\binom{k+t}{t}\cdot (n+m) \cdot t\cdot (k+t))$.
\end{lemma}

\begin{lemma}[\cite{ChitnisCHPP16}]\label{lem:derand}
Given a set $U$ of size $n$ and integers $0\leq a,b\leq n$, one can construct in time $2^{\Oh(\min\{a,b\}\log (a+b))}n\log n$ a family $\mathcal{S}$ of at most  $2^{\Oh(\min\{a,b\}\log (a+b))}\log n$ subsets of $U$ such that the following holds: for any sets $A,B\subseteq U$, $A\cap B=\emptyset$, $|A|\leq a$, $|B|\leq b$, there exists a set $S\in \mathcal{S}$ with $A\subseteq S$ and $B\cap S=\emptyset$. 
\end{lemma}

\begin{theorem}\label{thm:CSWCS} If property $\Pi$ can be recognized in time $f(n)$, then
\probCSWCS can be solved both in time $2^{k+t}\cdot f(k)\cdot n^{\Oh(1)}$, and in time $2^{\Oh(\min\{k,t\}\log (k+t))}\cdot f(k)\cdot n^{\Oh(1)}$.
\end{theorem}

\begin{proof}
Let $(G,c,\omega,k,t,w)$ be an instance of \probCSWCS.

First, we use Lemma~\ref{lem:conn-subgr} and in time $2^{k+t}\cdot n^{\Oh(1)}$ enumerate all connected $U\subseteq V(G)$ with $|U|=k$ and $|N_G(U)|\leq t$. By Lemma~\ref{lem:conn-subgr}, we have at most  $\binom{k+t}{t}tn$ sets. For every such a set $U$, we check in time $f(k)+O(k)$ whether the colored induced subgraph $G[U]$ satisfies $\Pi$ and $\omega(U)\geq w$. It is straightforward  to see that  $(G,c,\omega,k,t,w)$ is a yes-instance if and only if we find $U$ with these properties.

To construct the second algorithm, assume that  $(G,c,\omega,k,t,w)$ is a yes-instance. Then there is $U\subseteq V(G)$ such that $U$ is a connected $k$-vertex set such that $|N_G(U)|\leq t$, $\omega(U)\geq w$ and the colored graph $H=G[U]$ satisfies $\Pi$. Using Lemma~\ref{lem:derand}, we can construct in time $2^{\Oh(\min\{k,t\}\log (k+t))}\cdot n^{\Oh(1)}$ a family $\mathcal{S}$ of at most  $2^{\Oh(\min\{k,t\}\log (k+t))}\log n$ subsets of $V(G)$ such that the following holds: for any sets $A,B\subseteq V(G)$, $A\cap B=\emptyset$, $|A|\leq k$, $|B|\leq t$, there exists a set $S\in \mathcal{S}$ with $A\subseteq S$ and $B\cap S=\emptyset$. In particular, we have that there is $S\in \mathcal{S}$ such that $U\subseteq S$ and $N_G(U)\cap S=\emptyset$. It implies that $G[U]$ is a component of $G[S]$. 

Therefore, $(G,c,\omega,k,t,w)$ is a yes-instance if and only if there is $S\in\mathcal{S}$ such that a component of $G[S]$ is a solution for the instance. We construct the described set $\mathcal{S}$. Then for every $S\in\mathcal{S}$, we consider the components of  $G[S]$, and for every component $H$, we verify in time $f(k)+\Oh(k)$, whether $H$ gives us a solution. 
\end{proof}

Theorem~\ref{thm:CSWCS} immediately gives the following corollary.

\begin{corollary}\label{cor:CSWCS}
If $\Pi$ can be recognized in polynomial time, then \probCSWCS can be solved both in time $2^{k+t}\cdot n^{\Oh(1)}$, and in time $2^{\Oh(\min\{k,t\}\log (k+t))}\cdot n^{\Oh(1)}$.
\end{corollary}

\section{\probLCSS for properties characterized by finite sets of forbidden induced subgraphs}\label{sec:fpt-forb}
In this section we show that \probLCSS is \classFPT parameterized by $t$ when $\Pi$ is characterized by a finite set of forbidden induced subgraphs. 
We refer to this restriction of our problem as \probLForb.
Throughout this section, we assume that we are given a fixed finite set $\mathcal{F}$ of graphs. 

Recall that to apply the recursive understanding technique introduced by Chitnis et al.~\cite{ChitnisCHPP16}, we should be able to recurse when the input graph contains a separator of bounded size that separates the graph into two sufficiently big parts. To do this, we have to combine partial solutions in both parts. A danger in our case is that a partial solution in one part might contain a subgraph of a graph in $\mathcal{F}$.
We have to avoid creating subgraphs belonging to $\mathcal{F}$ when we combine partial solutions.
To achieve this goal, we need some definitions and auxiliary combinatorial results.

Let $p$ be a nonnegative integer. A pair $(G,x)$, where $G$ is a graph and $x=(x_1,\ldots,x_p)$ is a $p$-tuple of distinct vertices of $G$, is called a \emph{$p$-boundaried graph} 
or simply a \emph{boundaried graph}. Respectively, $x=(x_1,\ldots,x_p)$ is a \emph{boundary}. Note that a boundary is an ordered set. Hence, two $p$-boundaried graphs that differ only by the order of the vertices in theirs boundaries are distinct. Observe also that a boundary could be empty. We say that $(G,x)$ is a \emph{properly $p$-boundaried graph} if each component of $G$ has at least one vertex of the boundary. Slightly abusing notation, we may say that $G$ is a ($p$-) boundaried graph assuming that a boundary is given.

Two $p$-boundaried  graphs $(G_1,x^{(1)})$ and $(G_2,x^{(2)})$, where $x^{(h)}=(x_1^{(h)},\ldots,x_p^{(h)})$ for $h=1,2$, are \emph{isomorphic} if there is an isomorphism of $G_1$ to $G_2$ that maps each $x_i^{(1)}$ to $x_i^{(2)}$ for $i\in\{1,\ldots,p\}$. 
We say that $(G_1,x^{(1)})$ and $(G_2,x^{(2)})$ are \emph{boundary-compatible} if  for any distinct  $i,j\in\{1,\ldots,p\}$, $x_i^{(1)}x_j^{(1)}\in E(G_1)$ if and only if $x_i^{(2)}x_j^{(2)}\in E(G_2)$.

Let  $(G_1,x^{(1)})$ and $(G_2,x^{(2)})$ be  boundary-compatible $p$-boundaried graphs and let  $x^{(h)}=(x_1^{(h)},\ldots,x_p^{(h)})$ for $h=1,2$.  We define the \emph{boundary sum} $(G_1,x^{(1)})\oplus_b(G_2,x^{(2)})$ (or  simply $G_1\oplus_b G_2$) as the (non-boundaried) graph obtained by taking disjoint copies of $G_1$ and $G_2$ and identifying $x_i^{(1)}$ and $x_i^{(2)}$ for each $i\in\{1,\ldots,p\}$.

 Let $G$ be a graph and let $y=(y_1,\ldots,y_p)$ be a $p$-tuple of vertices of $G$. For a $s$-boundaried graph $(H,x)$ with the boundary $x=(x_1,\ldots,x_s)$ and pairwise distinct $i_1,\ldots,i_s\in \{1,\ldots,p\}$, we say that $H$ is an \emph{induced boundaried subgraph of $G$ with respect to $(y_{i_1},\ldots,y_{i_s})$} if $G$ contains an induced subgraph $H'$ isomorphic to $H$ such that the corresponding isomorphism of $H$ to $H'$ maps $x_j$ to $y_{i_j}$ for $j\in\{1,\ldots,s\}$ and $V(H')\cap \{y_1,\ldots,y_p\}=\{y_{i_1},\ldots,y_{i_s}\}$.

We construct the set of boundaried graphs $\mathcal{F}_b$ as follows. For each $F\in \mathcal{F}$, each separation $(A,B)$ of $F$ and each $p=|A\cap B|$-tuple $x$ of the vertices of $(A\cap B)$, we include $(F[A],x)$ in $\mathcal{F}_b$ unless it already contains an isomorphic boundaried graph. 
We say that two properly $p$-boundaried 
graphs $(G_1,x^{(1)})$ and $(G_2,x^{(2)})$, where $x^{(h)}=(x_1^{(h)},\ldots,x_p^{(h)})$, are \emph{equivalent  (with respect to $\mathcal{F}_b$)} if 
\begin{itemize}
\item[(i)] $(G_1,x^{(1)})$ and $(G_2,x^{(2)})$ are boundary-compatible,
\item[(ii)] for any $i,j\in\{1,\ldots,p\}$, $x_i^{(1)}$ and $x_j^{(1)}$ are in the same component of $G_1$ if and only if $x_i^{(2)}$ and $x_j^{(2)}$ are in the same component of $G_2$,
\item[(iii)] for any pairwise distinct $i_1,\ldots,i_s\in\{1,\ldots,p\}$, $G_1$ contains an $s$-boundaried induced subgraph $H\in \mathcal{F}_b$
 with respect to the $s$-tuple $(x_{i_1}^{(1)},\ldots,x_{i_s}^{(1)})$ if and only if   $H$ is an  induced subgraph of $G_2$ with respect to the $s$-tuple $(x_{i_1}^{(2)},\ldots,x_{i_s}^{(2)})$.
\end{itemize}
It is straightforward to verify that the introduced relation is indeed an equivalence relation on the set of properly $p$-boundaried  graphs. 
The following property of the equivalence with respect to $\mathcal{F}_b$ is crucial for our algorithm.

\begin{lemma}\label{lem:eq}
Let $(G,x)$, $(H_1,y^{(1)})$ and $(H_2,y^{(2)})$ be  boundary-compatible $p$-boundaried graphs, $x=(x_1,\ldots,x_p)$ and $y^{(h)}=(y_1^{(h)},\ldots,y_p^{(h)})$ for $h=1,2$. 
If $(H_1,y^{(1)})$ and $(H_2,y^{(2)})$ are equivalent with respect to $\mathcal{F}_b$, then $(G,x)\oplus_b (H_1,y^{(1)})$ is $\mathcal{F}$-free if and only if $(G,x)\oplus_b (H_2,y^{(2)})$ is $\mathcal{F}$-free.
\end{lemma}

\begin{proof}
By symmetry, it is sufficient to show that if $G\oplus_b H_1$ is not  $\mathcal{F}$-free, then the same hold for $G\oplus_b H_2$. 
Suppose that $F$ is an induced subgraph of $G\oplus_b H_1$ isomorphic to a graph of $\mathcal{F}$.
If $V(F)\subseteq V(G)$, then the claim is trivial. Suppose that this is not the case and $V(F)\cap V(H_1)\neq\emptyset$. Recall that $G\oplus_b H_1$ is obtained by identifying each $x_i$ and $y_i^{(1)}$. Denote the identified vertices by $y_1^{(1)},\ldots, y_p^{(1)}$.
 Let $F_1=F[V(F)\cap V(H_1)]$ and $F'=F[V(F)\cap V(G)]$; note that $F'$ could be empty. 
 Let $\{y_{i_1}^{(1)},\ldots,y_{i_s}^{(1)}\}=V(F)\cap \{y_1^{(1)},\ldots, y_p^{(1)}\}$; note that this set could be empty. 
Clearly, $(F_1,(y_{i_1}^{(1)},\ldots,y_{i_s}^{(1)}))$ is an $s$-boundaried subgraph of $H_1$ with respect to $(y_1^{(1)},\ldots,y_p^{1})$.
Observe that $\mathcal{F}_b$ contains an $s$-boundaried graph isomorphic to $(F_1,(y_{i_1}^{(1)},\ldots,y_{i_s}^{(1)}))$. Because $H_1$ and $H_2$ are equivalent with respect to $\mathcal{F}_b$, there is an induced $s$-boundaried subgraph
$(F_2,(y_{i_1}^{(2)},\ldots,y_{i_s}^{(2)}))$ of $H_2$ with respect to $(y_1^{(2)},\ldots,y_p^{(2)})$ isomorphic to $(F_1,(y_{i_1}^{(1)},\ldots,y_{i_s}^{(1)}))$. Then $F'\oplus_b F_2$ is isomorphic to $F$, that is, $G\oplus_b H_2$ contains $F$ as an induced subgraph.
\end{proof}

\begin{lemma}\label{lem:eq-check}
It can be checked in time $(|V(G_1)|+|V(G_2)|)^{\Oh(1)}$ whether two properly $p$-boundaried graphs $G_1$ and $G_2$ are equivalent, and the constant hidden in the $\Oh$-notation depends on $\mathcal{F}$ only. 
\end{lemma}

\begin{proof}
Let $(G_1,x^{(1)})$ and $(G_2,x^{(2)})$, where $x^{(h)}=(x_1^{(h)},\ldots,x_p^{(h)})$ for $h=1,2$, be two boundaried graphs. Clearly, the conditions (i) and (ii) of the definition of the equivalence with respect to $\mathcal{F}$ can be checked in polynomial time. 
To verify (iii), let $a=|\mathcal{F}_b|$, $b$ be the maximum size of the boundary of graphs in $\mathcal{F}_b$ and let $c$ be the maximum number of vertices of a graph in $\mathcal{F}_b$. 
Clearly, the values of $a$, $b$ and $c$ depend on $\mathcal{F}$ only. 
For each $s$-tuple of indices $(i_1,\ldots,i_s)$ where $s\leq b$, we check whether an $s$-boundaried graph $H\in \mathcal{F}_b$ is an $s$-boundaried induced subgraph of $G_1$ and $G_2$ with respect to $(x_{i_1}^{(1)},\ldots,x_{i_s}^{(1)})$ and $(x_{i_1}^{(2)},\ldots,x_{i_s}^{(2)})$ respectively. Since there are at most $bp^b$ $s$-tuples of indices $(i_1,\ldots,i_s)$, at most $a$ graphs
in $\mathcal{F}_b$ and  $G_h$ has at most $c|V(G_h)|^c$ induced subgraphs with at most $c$ vertices for $h=1,2$, we have that (iii) can be checked in polynomial time.
\end{proof}

For each nonegative integer $p$, we consider a set $\mathcal{G}_p$ of properly $p$-boundaried  graphs obtained by picking a graph with minimum number of vertices in each equivalence class.  
We show that the size of $\mathcal{G}_p$ and the size of each graph in the set is upper bounded by some functions of $p$, and this set  can be constructed in time that depends only on $p$ assuming that $\mathcal{F}_b$ is fixed.  
We need the following observation made by Fomin et al.~\cite{FominGKK16}.

\begin{lemma}[\cite{FominGKK16}]\label{lem:bigdeg}
Let $G$ be a connected graph and $S\subseteq V(G)$. Let $F$ be an inclusion minimal connected induced subgraph of $G$ such that $S\subseteq V(F)$ and $X=\{v\in V(F)|d_F(v)\geq 3\}\cup S$. Then $|X|\leq 4|S|-6$.
\end{lemma}

\begin{lemma}\label{lem:size-eq}
For every positive integer $p$, $|\mathcal{G}_p|=2^{\Oh(p^2)}$, and for every $H\in \mathcal{G}_p$, $|V(H)|=p^{\Oh(1)}$, where the constants hidden in the $\Oh$-notations depend on $\mathcal{F}$ only. 
Moreover, for every $p$-boundaried graph $G$, the number of $p$-boundaried graphs in $\mathcal{G}_p$ that are compatible with $G$ is $2^{\Oh(p\log p)}$.
\end{lemma}

\begin{proof}
Let $a=|\mathcal{F}_b|$, $b$ be the maximum size of the boundary of graphs in $\mathcal{F}_b$ and let $c$ be the maximum number of vertices of a graph in $\mathcal{F}_b$. 
Clearly, the values of $a$, $b$ and $c$ depend on $\mathcal{F}$ only. 
Assume that the boundary $x=(x_1,\ldots,x_p)$ is fixed. 

There are $2^{\binom{p}{2}}$ possibilities to select a set of edges with both end-vertices in $\{x_1,\ldots,x_p\}$. 
The number of possible partitions of the boundary into components is the Bell number $B_p$ and $B_p=2^{\Oh(p\log p)}$.
The number of $s$-tuples of vertices of $\{x_1,\ldots,x_p\}$ that could be boundaries of the copies of $s$-boundaried induced subgraphs  $H\in \mathcal{F}_b$  is at most $bp^b$. 
Respectively, the number of distinct equivalence classes is at most $2^{\binom{p}{2}}B_p bp^b2^a$, that is, $|\mathcal{G}_p|\leq 2^{\binom{p}{2}}B_pbp^b2^a=2^{\Oh(p^2)}$.  

Let $G$ be a $p$-boundaried graph in one of the classes with minimum number of vertices. Notice that for each  $s$-tuple of vertices of $\{x_1,\ldots,x_p\}$, $G$ can contain several copies of the same $H$ as an induced subgraph with respect to this $s$-tuple. In this case we pick one of them and obtain that  $G$ contains at most $bp^b2^a$ distinct boundaried induced subgraphs $H\in \mathcal{F}_b$. 
Let $S$ be the set of vertices of $G$ that belong to these subgraphs or to the boundary $x$. We have that $|S|\leq bp^b2^ac+p$. Let $X=\{v\in V(G)\mid d_G(v)\geq 3\}\cup S$. By applying Lemma~\ref{lem:bigdeg} to each component of $G$, we obtain that $|X|\leq 4|S|-6$.  

By the minimality of $G$, every vertex of degree one is in $S$. Hence, $Y=V(G)\setminus X$ contains only vertices of degree two and, therefore,  $G[Y]$ is a union of disjoint paths.
Observe that by the minimality of $G$, each vertex of $Y$ is a cut vertex of the component of $G$ containing it. It implies that $G[Y]$ contains at most $|X|-1$ paths. Suppose that $G[Y]$ contains a path $P$ with at least $2c+2$ vertices. Let $G'$ be the graph obtained from $G$ by the contraction of one edge of $P$. We claim that $G$ and $G'$ are equivalent with respect to $\mathcal{F}_b$. Since the end-vertices of the contracted edges are not the vertices of the boundary, the conditions (i) and (ii) of the definition of the equivalence are fulfilled. 
Therefore, it is sufficient to verify (iii). Let 
$i_1,\ldots,i_s\in\{1,\ldots,p\}$.  Suppose that $G$ contains an $s$-boundaried induced subgraph $H\in \mathcal{F}_b$ with respect to the $s$-tuple $(x_{i_1},\ldots,x_{i_s})$. Then at least two adjacent vertices of $P$ are not included in the copy of $H$ in $G$. It implies that 
$H$ is an  induced subgraph of $G'$ with respect to $(x_{i_1},\ldots,x_{i_s})$. Suppose that $G'$ contains an $s$-boundaried induced subgraph $H\in \mathcal{F}_b$ with respect to the $s$-tuple $(x_{i_1},\ldots,x_{i_s})$. Then at least one vertex of $P$ is not included in the copy of $H$ in $G'$. Then 
$H$ is an  induced subgraph of $G'$ with respect to $(x_{i_1},\ldots,x_{i_s})$. But the equivalence of $G$ and $G'$ contradicts the minimality of $G$. We conclude that each path in $G[Y]$ contains at most $2c+1$ vertices. Then the total number of vertices of $G$ is at most $|X|+(|X|-1)(2c+1)=p^{\Oh(1)}$.

To see that for any $p$-boundaried graph $G$, the number of  graphs in $\mathcal{G}_p$ that are compatible with $G$ is $2^{\Oh(p\log p)}$, notice that if $(H,(x_1,\ldots,x_p))\in \mathcal{G}_p$ and is compatible with $G$, then the adjacency between the vertices of the boundary is defined by $G$. 
Then the number of $s$-tuples of vertices of $\{x_1,\ldots,x_p\}$ that could be boundaries of the copies of $s$-boundaried induced subgraphs  from $\mathcal{F}_b$  is at most $bp^b$ and for each $s$-tuple we can have at most $2^a$ $s$-boundaried induced subgraphs  from $\mathcal{F}_b$.
Taking into account that there $2^{\Oh(p\log p)}$ possibilities for the verties of the boundary be partitioned according to their inclusions in the components, we obtain the claim.
\end{proof}

Consider now the class $\mathcal{C}$ of $p$-boundaried graphs, such that a $p$-boundaried graph \linebreak $(G,(x_1,\ldots,x_p\})$ $\in\mathcal{C}$ if and only if it holds that for every component $H$ of $G-\{x_1,\ldots,x_p\}$, $N_G(V(H))=\{x_1,\ldots,x_p\}$. We consider our equivalence relation with respect to $\mathcal{F}_b$ on $\mathcal{C}$ and define $\mathcal{G}_p'$ as follows. 
In each equivalence class, we select a graph  $(G,(x_1,\ldots,x_p\})\in\mathcal{C}$ such that both the number of components of $G-\{x_1,\ldots,x_p\}$ is minimum and the number of vertices of $G$ is minimum subject to the first condition, and then include it in $\mathcal{G}_p'$.

\begin{lemma}\label{lem:size-eq-prime}
For every positive integer $p$, $|\mathcal{G}_p'|=2^{\Oh(p^2)}$, and for each $H\in \mathcal{G}_p'$, $|V(H)|=p^{\Oh(1)}$, and the constants hidden in the $\Oh$-notations depend on $\mathcal{F}$ only. 
Moreover, for any $p$-boundaried graph $G$, the number of $p$-boundaried graphs in $\mathcal{G}_p'$ that are compatible with $G$ is $p^{\Oh(1)}$.
\end{lemma}

\begin{proof}
Let $a=|\mathcal{F}_b|$, $b$ be the maximum size of the boundary of graphs in $\mathcal{F}_b$ and let $c$ be the maximum number of vertices of a graph in $\mathcal{F}_b$. 
Clearly, the values of $a$, $b$ and $c$ depend on $\mathcal{F}$ only. 
Assume that the boundary $x=(x_1,\ldots,x_p)$ is fixed. 

There are $2^{\binom{p}{2}}$ possibilities to select a set of edges with both end-vertices in $\{x_1,\ldots,x_p\}$. 
The number of $s$-tuples of vertices of $\{x_1,\ldots,x_p\}$ that could be boundaries of the copies of $s$-boundaried induced subgraphs  $H\in \mathcal{F}_b$  is at most $bp^b$. 
Respectively, the number of distinct equivalence classes of $\mathcal{C}$ is at most $2^{\binom{p}{2}} bp^b2^a$, that is, $|\mathcal{G}_p'|\leq 2^{\binom{p}{2}}bp^b2^a=2^{\Oh(p^2)}$.  

Let $(G,x)$ be a $p$-boundaried graph in one of the classes such that the number of components of $G-\{x_1,\ldots,x_p\}$ is minimum and the number of vertices of $G$ is minimum subject to (i). Let $Q_1,\ldots,Q_r$ be the components of $G-\{x_1,\ldots,x_p\}$. Let $Q_i'=G[V(Q_i)\cup\{x_1,\ldots,x_p\}]$.

Let $i\in\{1,\ldots,r\}$. Let also $Q_i''$ be the graph obtained from $Q_i$ by the deletion of the edges with both end-vertices in the boundary.
Observe that for each  $s$-tuple of vertices of $\{x_1,\ldots,x_p\}$, $Q_i'$ can contain several copies of the same $H$ as an induced subgraph with respect to this $s$-tuple. In this case we pick one of them and obtain that  $Q_i'$ contains at most $bp^b2^a$ distinct boundaried induced subgraphs $H\in \mathcal{F}_b$. 
Let $S$ be the set of vertices of $Q_i'$ that belong to these subgraphs or to the boundary $x$. We have that $|S|\leq bp^b2^ac+p$. Let $X=\{v\in V(G)\mid d_{Q_i''}(v)\geq 3\}\cup S$. By applying Lemma~\ref{lem:size-eq} to  $Q_i''$, we obtain that $|X|\leq 4|S|-6$.  Then by the same arguments as in the proof of Lemma~\ref{lem:size-eq}, we obtain that 
$Q_i''$ and, therefore, $Q_i'$ has at most $|X|+(|X|-1)(2c+1)=p^{\Oh(1)}$ vertices. Since $V(Q_i)\subseteq V(Q_i')$, we have that $Q_i$ has $p^{\Oh(1)}$ vertices.

Suppose that there are $c+1$ pairwise distinct but equivalent $(Q_{j_0}',x),\ldots,(Q_{j_c}',x)$ for $j_0,\ldots,j_c\in\{1,\ldots,r\}$. 
Assume that $G$ contains an $s$-boundaried induced subgraph $H\in \mathcal{F}_b$
 with respect to an $s$-tuple $(x_{i_1},\ldots,x_{i_s})$ for some $i_1,\ldots,i_s\in\{1,\ldots,p\}$. 
Since $|V(H)|\leq c$, there is $h\in\{0,\ldots,c\}$ such that $V(H)\cap V(Q_{j_h})=\emptyset$. Because $(Q_{j_0}',x),\ldots,(Q_{j_c}',x)$ are equivalent, we obtain that $H$ is an $s$-bounderied induced subgraph of $G-V(Q_{j_0})$ contradicting the condition (i) of the choice of $G$.
Therefore, there are at most $c$ pairwise equivalent boundaried graphs in $\{Q_1',\ldots,Q_r'\}$.

We claim that the number of pairwise nonequivalent graphs in  $\{Q_1',\ldots,Q_r'\}$ is $p^{\Oh(1)}$. 
Notice that the adjacency between the boundary vertices is defined by $G$. 
Then the number of $s$-tuples of vertices of $\{x_1,\ldots,x_p\}$ that could be boundaries of the copies of $s$-boundaried induced subgraphs  from $\mathcal{F}_b$  is at most $bp^b$ and for each $s$-tuple we can have at most $2^a$ $s$-boundaried induced subgraphs  from $\mathcal{F}_b$. Then the claim follows.

We conclude that $r=cp^{\Oh(1)}$. Since $|V(Q_i)|=p^{\Oh(1)}$ for each $i\in\{1,\ldots,r\}$, $|V(G)|=p^{\Oh(1)}$.

To see that for any $p$-boundaried graph $G$, the number of  graphs in $\mathcal{G}_p'$ that are compatible with $G$ is $p^{\Oh(1)}$, notice that if $(H,(x_1,\ldots,x_p))\in \mathcal{G}_p$ and is compatible with $G$, then the adjacency between the vertices of the boundary is defined by $G$. 
Then the number of $s$-tuples of vertices of $\{x_1,\ldots,x_p\}$ that could be boundaries of the copies of $s$-boundaried induced subgraphs  from $\mathcal{F}_b$  is at most $bp^b$ and for each $s$-tuple we can have at most $2^a$ $s$-boundaried induced subgraphs  from $\mathcal{F}_b$.
\end{proof}

Lemmas~\ref{lem:eq-check}, \ref{lem:size-eq} and \ref{lem:size-eq-prime} immediately imply that $\mathcal{G}_p$  and $\mathcal{G}_p'$ 
can be constructed by  brute force.

\begin{lemma}\label{lem:eq-constr}
The sets $\mathcal{G}_p$  and $\mathcal{G}_p'$
 can be constructed in time $2^{p^{\Oh(1)}}$. 
\end{lemma}

To apply the recursive understanding technique, we also have to solve a special variant of \probLCSS tailored for recursion. First, we define the following auxiliary problem for a positive integer $w$.

 \defproblema{\probMAXF}%
{A graph $G$, sets $I,O,B\subseteq V(G)$ such that $I\cap O=\emptyset$ and $I\cap B=\emptyset$, a weight function $\omega\colon V(G)\rightarrow\mathbb{Z}_{\geq 0}$ and a nonnegative integer $t$.}%
{Find a $t$-secluded $\mathcal{F}$-free induced connected subgraph $H$ of $G$ of maximum weight or weight at least $w$
such that $I\subseteq V(H)$, $O\subseteq V(G)\setminus V(H)$ and $N_G(V(H))\subseteq B$ and output $\emptyset$ if such a subgraph does not exist.}

Notice that \probMAXF is an optimization problem and a \emph{solution} is either an induced subgraph $H$ of maximum weight or of weight at least $w$, or $\emptyset$. Observe also that we allow zero weights for technical reasons.

We recurse if we can separate graphs by a separator of bounded size into two big parts and we use the vertices of the separator to combine partial solutions in both parts. This leads us to the following problem. Let $(G,I,O,B,\omega,t)$ be an instance of \probMAXF and let $T\subseteq V(G)$ be a set of \emph{border terminals}. We say that an instance $(G',I',O',B',\omega',t')$ is obtained by a \emph{border complementation} if  there is a partition $(X,Y,Z)$ of $T$ (some sets could be empty), where $X=\{x_1,\ldots,x_p\}$, such that $Y=\emptyset$ if $X=\emptyset$, $I\cap T\subseteq X$, $O\cap T\subseteq Y\cup Z$ and $Y\subseteq B$, and there is a $p$-boundaried graph $(H,y)\in \mathcal{G}_p$ such that $(H,y)$ and $(G,(x_1,\ldots,x_p))$ are boundary-compatible, and the following holds:
\begin{itemize} 
\item[(i)] $G'$ is obtained from $(G,(x_1,\ldots,x_p))\oplus_b (H,y)$ (we keep the notation $X=\{x_1,\ldots,x_p\}$ for the set of vertices obtained by the identification in the boundary sum) by adding edges joining every vertex of $V(H)$ with every vertex of $Y$,
\item[(ii)] $I'= I\cup V(H)$,
\item[(iii)] $O'=O\cup Y\cup Z$,
\item[(iv)] $B'=B\setminus X$,
\item[(v)] $\omega'(v)=\omega(v)$ for $v\in V(G)$ and $\omega'(v)=0$ for $v\in V(H)\setminus X$,
\item[(vi)] $t'\leq t$.
\end{itemize}
We also say that $(G',I',O',B',w',t')$ is a \emph{border complementation of $(G,I,O,B,w,t)$ with respect to $(X=\{x_1,\ldots,x_p\},Y,Z,H)$}.
We say that $(X=\{x_1,\ldots,x_p\},Y,Z,H)$ is \emph{feasible} if it holds that $Y=\emptyset$ if $X=\emptyset$, $I\cap T\subseteq X$, $O\cap T\subseteq Y\cup Z$ and $Y\subseteq B$, and the $p$-boundaried graph $H\in \mathcal{G}_p$  and $(G,(x_1,\ldots,x_p))$ are boundary-compatible.

 \defproblema{\probBORDF}%
{A graph $G$, sets $I,O,B\subseteq V(G)$ such that $I\cap O=\emptyset$ and $I\cap B=\emptyset$, a weight function $\omega\colon V(G)\rightarrow\mathbb{Z}_{\geq 0}$, a nonnegative integer $t$, and a set $T\subseteq V(G)$ of border terminals of size at most $2t$.}
{Output a solution for each instance $(G',I',O',B',w',t')$ of \probMAXF that can be obtained from $(G,I,O,B,w,t)$ by a border complementation distinct from the border complementation with respect to $(\emptyset,\emptyset,T,\emptyset)$, and for the border complementation with respect to $(\emptyset,\emptyset,T,\emptyset)$ output a nonempty solution if it has weight at least $w$ and output $\emptyset$ otherwise.}

Two instances $(G_1,I_1,O_1,B_1,\omega_1,t,T)$ and $(G_2,I_2,O_2,B_2,\omega_2,t,T)$ of \probBORDF (note that $t$ and $T$ are the same) are said to be \emph{equivalent} if 
\begin{itemize}
\item[(i)] $T\cap I_1=T\cap I_2$, $T\cap O_1=T\cap O_2$ and $T\cap B_1=T\cap B_2$,
\item[(ii)] for the border complementations $(G_1',I_1',O_1',B_1',\omega_1',t')$ and $(G_2',I_2',O_2',B_2',\omega_2',t')$ of the instances $(G_1,I_1,O_1,B_1,\omega_1,t')$ and $(G_2,I_2,O_2,B_2,\omega_2,t')$ respectively of \probMAXF with respect to every feasible $(X=\{x_1,\ldots,x_p\},Y,Z,H)$ and $t'\leq t$, it holds that
\begin{itemize}
\item[a)] if $(G_1',I_1',O_1',B_1',\omega_1',t')$ has a nonempty solution $R_1$, then $(G_2',I_2',O_2',B_2',\omega_2',t')$ has a nonempty solution $R_2$ with $w_2'(V(R_2))\geq \min\{\omega_1'(V(R_1)),w\}$ and, vice versa, 
\item[b)]if $(G_2',I_2',O_2',B_2',\omega_2',t')$ has a nonempty solution $R_2$, then $(G_1',I_1',O_1',B_1',\omega_1',t')$ has a nonempty solution $R_1$ with $\omega_1'(V(R_1))\geq \min\{\omega_2'(V(R_2)),w\}$.
\end{itemize}
\end{itemize}
Strictly speaking, if $(G_1,I_1,O_1,B_1,\omega_1,t,T)$ and $(G_2,I_2,O_2,B_2,\omega_2,t,T)$ are equivalent, then a solution of the first problem is not necessarily a solution of the second. Nevertheless,  \probBORDF is an auxiliary problem and in the end we use it to solve an instance $(G,\omega,t,w)$ of \probLForb by calling the algorithm for \probBORDF for $(G,\emptyset,\emptyset,V(G),\omega,t,\emptyset)$.
Clearly, $(G,\omega,t,w)$ is a yes-instance if and only if a solution for the corresponding instance of \probBORDF  contains a connected subgraph $R$ with $\omega(V(R))\geq w$.
It allows us to not distinguish equivalent instances of \probBORDF and their solutions. 

\subsection{High connectivity phase}
In this section we solve \probBORDF for $(q,t)$-unbreakable graphs. For this purpose, we use \emph{important separators} defined by Marx in~\cite{Marx06}. Essentially, we follow the terminology given in~\cite{CyganFKLMPPS15}.
Recall that for $X,Y\subseteq V(G)$, a set $S\subseteq V(G)$ is an $(X,Y)$-separator if $G-S$ has no path joining a vertex of $X\setminus S$ with a vertex of $Y\setminus S$. 
An $(X,Y)$-separator $S$ is \emph{minimal} if no proper subset of $S$ is an  $(X,Y)$-separator. For $X\subseteq V(G)$ and $v\in V(G)$, it is said that $v$ is \emph{reachable} from $X$ if there is an $(x,v)$-path in $G$ with $x\in X$. A minimal $(X,Y)$-separator $S$ can be characterized by the set of vertices reachable from $X\setminus S$ in $G-S$.

\begin{lemma}[\cite{CyganFKLMPPS15}]\label{lem:min-sep}
If $S$ is a minimal $(X,Y)$-separator in $G$, then $S=N_G(R)$ where $R$ is the set of vertices reachable from $X\setminus S$ in $G-S$. 
\end{lemma}

Let $X,Y\subseteq V(G)$ for a graph $G$. Let $S\subseteq V(G)$ be an $(X,Y)$-separator and let $R$ be the set of vertices reachable from $X\setminus S$ in $G-S$. It is said that 
$S$ is an \emph{important $(X,Y)$-separator} if $S$ is minimal and there is no $(X,Y)$-separator $S'\subseteq V(G)$ with $|S'|\leq |S|$ such that $R\subset R'$ where $R'$ is the set of vertices reachable from $X\setminus S'$ in $G-S'$.

\begin{lemma}[\cite{CyganFKLMPPS15}]\label{lem:imp-sep}
Let $X,Y\subseteq V(G)$ for a graph $G$, let $t$ be a nonnegative integer and let $\mathcal{S}_t$ be the set of all important $(X,Y)$-separators of size at most $t$. Then $|\mathcal{S}_t|\leq 4^t$
and $\mathcal{S}_t$ can be constructed in time $\Oh(|\mathcal{S}_t|\cdot t^2\cdot (n+m))$.
\end{lemma}

The following lemma shows that we can separately lists all graphs $R$ in a solution of \probBORDF with  $|V(R)\cap V(G)|\leq q$  and all graphs $R$ with $|V(G)\setminus V(R)|\leq q+t$.

\begin{lemma}\label{lem:sol-size}
Let $(G,I,O,B,\omega,t,T)$ be an instance of \probBORDF where $G$ is a $(q,t)$-unbreakable graph for a positive integer $q$. Then  for each nonempty graph $R$ in a solution of \probBORDF, either $|V(R)\cap V(G)|\leq q$ or  $|V(G)\setminus V(R)|\leq q+t$.
\end{lemma}

\begin{proof}
Let $R$ be a nonempty graph listed in a solution of \probBORDF for an instance $(G',I',O',B',\omega',t')$ \probMAXF.
Assume that $G'$ is obtained from $(G,(x_1,\ldots,x_p))\oplus_b (H,y)$ for $H\in \mathcal{G}_p$.
Let $U=N_{G'}[V(R)\cap V(G)]$ and $W=V(G)\setminus V(R)$. Clearly $(U,W)$ is a separation of $G$ of order at most $t$. Since $G$ is $(q,t)$-unbreakable, either $|U\setminus W|\leq q$ or $|W\setminus U|\leq q$.
If $|U\setminus W|\leq q$, then $|V(R)\cap V(G)|\leq |U\setminus W|\leq q$. 
If $|W\setminus U|\leq q$, then $|V(G)\setminus V(R)|\leq q+t$.
\end{proof}

Now we can prove the following crucial lemma. 

\begin{lemma}\label{lem:unbreak}
\probBORDF for $(q,t)$-unbreakable graphs can be solved in time  $2^{(q+t\log(q+t)))}\cdot n^{\Oh(1)}$ if the sets $\mathcal{G}_p$ for $p\leq 2t$  are given. 
\end{lemma}

\begin{proof}
Assume that the sets $\mathcal{G}_p$ for $p\leq 2t$  are given. We consider all possible instances $(G',I',O',B',\omega',t')$ of \probMAXF obtained from the input instance $(G,I,O,B,w,t)$ as it is explained in the definition of \probBORDF. To construct each instance, we consider all at most $3^{2t}$ partitions $(X,Y,Z)$ of $T$,   where $X=\{x_1,\ldots,x_p\}$, such that $Y=\emptyset$ if $X=\emptyset$, $I\cap T\subseteq X$, $O\cap T\subseteq Y\cup Z$ and $Y\subseteq B$. Then 
we consider all  $p$-boundaried graph $(H,y)\in \mathcal{G}_p$ such that $(H,y)$ and $(G,(x_1,\ldots,x_p))$ are boundary-compatible. By Lemma~\ref{lem:size-eq}, there are  $2^{\Oh(t\log t)}$ such sets.  Consider now a constructed instance $(G',I',O',B',\omega',t')$ and assume that $G'$ is obtained  from $(G,(x_1,\ldots,x_p))\oplus_b (H,y)$ for $(H,y)\in\mathcal{G}_p$.
We find a $t$-secluded $\mathcal{F}$-free induced connected subgraph $R$ of $G'$ of maximum weight 
such that $I'\subseteq V(R)$, $O'\subseteq V(G')\setminus V(R)$ and $N_{G'}(V(R))\subseteq B'$ if such a subgraph exists.
By Lemma~\ref{lem:sol-size}, either $|V(R)\cap V(G)|\leq q$ or  $|V(G)\setminus V(R)|\leq q+t$.

First, we find  
a $t$-secluded $\mathcal{F}$-free induced connected subgraph $R$ of $G'$ with $|V(R)\cap V(G)|\leq q$ of maximum weight 
such that $I'\subseteq V(R)$, $O'\subseteq V(G')\setminus V(R)$ and $N_{G'}(V(R))\subseteq B'$.
If  $|V(R)\cap V(G)|\leq q$, then  $|V(R)|\leq |V(H)|+|V(R)\cap V(G)|$. By Lemma~\ref{lem:size-eq}, $|V(H)|=t^{c}$ for some constant $c$. It implies that $|V(R)|\leq t^c+q$. To find $R$, we consider all $k\leq  t^c+q$ and find  a $t$-secluded $\mathcal{F}$-free induced connected subgraph $R$ of $G'$ of maximum weight such that $I'\subseteq V(R)$, $O'\subseteq V(G')\setminus V(R)$, $N_{G'}(V(R))\subseteq B'$ and $|V(R)|=k$. By Corollary~\ref{cor:CSWCS},
it can be done in time $2^{\Oh(\min\{t^c+q,t\}\log (t^c+q+t))}\cdot n^{\Oh(1)}$.

Now we find a $t$-secluded $\mathcal{F}$-free induced connected subgraph $R$ of $G'$ with $|V(G)\setminus V(R)|\leq q+t$ of maximum weight 
such that $I'\subseteq V(R)$, $O'\subseteq V(G')\setminus V(R)$ and $N_{G'}(V(R))\subseteq B'$.

Because $|V(R)\cap V(G)|\leq q+t$,
 there is a set $O'\subseteq S\subseteq V(G)\setminus V(R)$ such that $|S|\leq q+t$ and $G'-S$ is $\mathcal{F}$-free. 
We list all such sets $S$ 
using the standard branching algorithm for this problem (see, e.g.,~\cite{CyganFKLMPPS15}). The main idea of the algorithm is that if $G'$ has an induced subgraph $F$ isomorphic to a graph of $\mathcal{F}$, then at least one vertex of $F$ should be in $S$.  Initially we set $S=O$ and set a branching parameter $h=q+t-|O|$. If $h<0$, we stop.
We check whether $G'-S$ has an induced subgraph $F$ isomorphic to a graph of $\mathcal{F}$. If we have no such graph, we return $S$. If $V(F)\subseteq I'$, then we  stop.
Otherwise, we branch on the vertices of $F$. For each $v\in V(F)\setminus I'$, we set $S=S\cup \{v\}$ and set $h=h-1$ and recurse. It is straightforward to verify the correctness of the algorithm and see that it runs in in time $2^{\Oh(q+t)}\cdot n^{(1)}$, because $\mathcal{F}$ is fixed and each graph from this set has a constant size. If the algorithm fails to output any set $S$, then we conclude that $(G',I',O',B',\omega',t')$ has no solution $R$ with $|V(R)\setminus V(G)|\leq q$. From now on we assume that this is not the case. 

For each $S$, we set $O'=O'\cup S$  and find a solution for the modified instance $(G',I',O',B',\omega',t')$. Then we choose a solution of maximum size (if exist). 

If $I'=\emptyset$, we guess a vertex $u\in V(G')\setminus O'$ that is included in a solution. We set $I'=\{u\}$ and $B'=B'\setminus \{u\}$ and solve the modified instance $(G',I',O',B',w',t')$. Then we choose a solution of maximum size for all guesses of $u$. From now on we have $I'\neq \emptyset$.

We apply a series of reduction rules for $(G',I',O',B',\omega',t')$. Let $h=q+t$.

\begin{reduction}\label{Ared:one}
If $G'$ is disconnected and has vertices of $I'$ in distinct components, then return the answer no and stop. Otherwise, let $Q$ be a component of $G'$ containing $I'$ and  set $G'=Q$, $B'=B'\cap V(Q)$, $O'=O'\cap V(Q)$ and $h=h-|V(G)\setminus V(Q)|$. If $h<0$, then return no and stop.  
\end{reduction}

It is straightforward to see that the rule is safe, because a solution is a connected graph. Notice that from now on we can assume that $G'$ is connected. If $O'=\emptyset$, $G'$ is a solution and we get the next rule.

\begin{reduction}\label{Ared:two}
If $O'=\emptyset$, then return $G'$.  
\end{reduction}

From now on we assume that $O'\neq\emptyset$.
Let $Q$ be a component of $G'-B'$. Notice that for any solution $R$, either $V(Q)\subseteq V(R)$ or $V(Q)\cap V(R)=\emptyset$, because $N_{G'}(V(R))\subseteq B'$. Moreover, if
$V(Q)\cap V(R)=\emptyset$, then $N_{G'}[V(Q)]\cap V(R)=\emptyset$.
This  leads to the following rule.

\begin{reduction}\label{Ared:three}
For a component $Q$ of $G'-B'$ do the following in the given order:
\begin{itemize}
\item if $V(Q)\cap I'\neq\emptyset$ and $V(Q)\cap O'\neq\emptyset$, then return no and stop,
\item if $V(Q)\cap I'\neq\emptyset$, then set $I'=I'\cup V(Q)$,
\item if $V(Q)\cap O'\neq\emptyset$, then set $O'=O'\cup N_{G'}[V(Q)]$.
\end{itemize}
\end{reduction}

The rule is applied for each component $Q$ exactly once.
Now our aim is find all inclusion maximal induced subgraphs $R$ of $G'$ such that $I'\subseteq V(R)$, $O'\cap V(R)=\emptyset$, $N_{G'}(V(R))\subseteq B'$, $|N_{G'}(V(R))|\leq t'$ and all the vertices of $R$ are reachable from $I'$. Then, by maximality, a solution is such a subgraph $R$ that is connected and, subject to connectivity, has a maximum weight. We doing it using important $(N_{G'}[I'],O')$-separators. The obstacle is the condition that $N_{G'}(V(R))\subseteq B'$.  Let $Q$ be a component of $G'-O'$. By Reduction Rule~\ref{Ared:three}, we have that exactly one of the following holds: either (i) $V(Q)\subseteq I'$ or (ii) $V(Q)\subseteq O'$ or (iii) $V(Q)\cap I'=V(Q)\cap O'=\emptyset$. To ensure that  $N_{G'}(V(R))\subseteq B'$, we have to insure that if (iii) is fulfilled, then it holds that either $V(Q)\subseteq V(R)$ or $V(Q)\cap V(R)=\emptyset$. To do it, we construct the auxiliary graph $G''$ as follows. For each $v\in V(G')\setminus (I'\cup O'\cup B')$, we replace $v$ by $t+1$ true twin vertices 
$v_0,\ldots,v_t$ that are adjacent to the same vertices as $v$ in $G$ or to the corresponding true twins obtained from the neighbors of $v$.  For an induced subgraph $R$ of $G'$, we say that the induced subgraph $R'$ of $G''$ is an \emph{image} of $R$ if $R'$ is obtained by the same replacement of the vertices  $v\in V(R)\setminus (I'\cup O'\cup B')$ by $t+1$ twins. Respectively, we say that $R$ is a \emph{preimage} of $R'$. 

We claim that if $R'$ is an  induced subgraph of $G''$ such that $I'\subseteq V(R')$, $O'\cap V(R')=\emptyset$,
$|N_{G''}(V(R'))|\leq t'$ and all the vertices of $R'$ are reachable from $I'$, then $R'$ has a preimage $R$ and $N_{G'}(V(R))=N_{G''}(V(R'))\subseteq B'$. 

To prove the claim, consider an inclusion maximal induced subgraph $R'$ of $G''$ such that $I'\subseteq V(R')$, $O'\cap V(R')=\emptyset$,  
$|N_{G''}(V(R'))|\leq t'$ and all the vertices of $R'$ are reachable from $I'$. Let $v'\in N_{G''}(R'')$ and assume that $v'\notin B'$. Clearly, $v'\notin I'$. Notice that $v'\notin O'$, 
because by Reduction Rule~\ref{Ared:three}, we have that for any $w\in O'\setminus B'$, $N_{G'}[w]\subseteq O'$. Since $v'\notin B'\cup I'\cup O'$, $v'\in\{v_0,\ldots,v_t\}$ for $t+1$ true twins constructed for some vertex $v\in V(G')$. Because $|N_{G''}(V(R'))|\leq t'$, there is $i\in\{0,\ldots,t\}$, such that $v_i\notin N_{G''}(R'')$.
As $v_i$ and $v'$ are twins, $v_i\in V(R'')$. Let $R''=G''[V(R')\cup \{v_0,\ldots,v_t\}]$. We obtain that $I'\subseteq V(R')$, $O'\cap V(R')=\emptyset$, 
$|N_{G''}(V(R'))|\leq t'$ and all the vertices of $R'$ are reachable from $I'$, but $V(R')\subset V(R'')$ contradicting maximality. Hence, $N_{G''}(V(R'))\subseteq B'$. Then $R'$ has a preimage $R$ and $N_{G'}(V(R))=N_{G''}(V(R'))\subseteq B'$. 

Using the  claim, we conclude that to find all inclusion maximal induced subgraphs $R$ of $G'$ such that $I'\subseteq V(R)$, $O'\cap V(R)=\emptyset$, $N_{G'}(V(R))\subseteq B'$, $|N_{G'}(V(R))|\leq t'$ and all the vertices of $R$ are reachable from $I'$, we should list inclusion maximal induced subgraphs $R'$ of $G''$ such that $I'\subseteq V(R')$, $O'\cap V(R')=\emptyset$,  
$|N_{G''}(V(R'))|\leq t'$ and all the vertices of $R'$ are reachable from $I'$, and then take preimages of the graphs $R'$.

To find the  maximal induced subgraphs $R'$ of $G''$ such that $I'\subseteq V(R')$, $O'\cap V(R')=\emptyset$,  
$|N_{G''}(V(R'))|\leq t'$ and all the vertices of $R'$ are reachable from $I'$, we use Lemma~\ref{lem:imp-sep}.
In time $4^t\cdot n^{\Oh(1)}$ we construct the set $\mathcal{S}_{t'}$ of all important $(N_{G''}[I'],O')$-separators of size at most $t'$ in $G''$. Then for each $S\in \mathcal{S}_{t'}$, we find $R'$ that is the union of the components of $G''-S$ containing the vertices of $I'$.

Since Reduction Rules~\ref{Ared:one}--\ref{Ared:three} can be applied in polynomial time and $G''$ can be constructed in polynomial time, we have that the total running time is 
$2^{\Oh(t+q)}\cdot n^{\Oh(1)}$.

Now we compare the two subgraphs $R$ that we found for the cases $|V(R)\cap V(G)|\leq q$ and $|V(G)\setminus V(R)|\leq q+t$ and output the subgraph of maximum weight or the empty set if we failed to find these subgraph. Taking into account the time used to construct the instances $(G',I',O',B',\omega',t')$, we obtain that the total running time is $2^{\Oh(q+t\log(q+t)))}\cdot n^{\Oh(1)}$
\end{proof}

\subsection{The \classFPT algorithm for \probLForb}\label{Asec:forb}
In this section we construct an \classFPT algorithm for \probLForb parameterized by $t$. We do this by solving \probBORDF in \classFPT-time for general case.

\begin{lemma}\label{lem:bordforb}
\probBORDF  can be solved in time $2^{2^{2^{\Oh(t\log t)}}}\cdot n^{\Oh(1)}$.
\end{lemma}

\begin{proof}
Given $\mathcal{F}$, we construct the set $\mathcal{F}_b$. Then we use Lemma~\ref{lem:eq-constr} to construct the sets $\mathcal{G}_p$ for $p\in \{0,\ldots,t\}$ in time $2^{t^{\Oh(1)}}$.

By Lemma~\ref{lem:size-eq}, there is a constant $c$ that depends only on $\mathcal{F}$ such that  for every nonnegative $p$ and 
for any $p$-boundaried graph $G$, there are at most $2^{c p\log p}$ $p$-boundaried graphs in $\mathcal{G}_p$ that are compatible with $G$ and 
there are at most $p^{c}$  $p$-boundaried graphs in $\mathcal{G}_p'$ that are compatible with $G$. 
We define 
\begin{equation}\label{Aeq:q}
q=2^{((t+1)t3^{2t}2^{c2t\log (2t)}+2t)}\cdot 2((t+1)t3^{2t}2^{c2t\log (2t)}+2t)^c t^c+(t+1)t3^{2t}2^{c2t\log (2t)}+2t.
\end{equation}
The choice of $q$ will become clear later in the proof. 
Notice that $q=2^{2^{\Oh(t\log t)}}$.

Consider an instance $(G,I,O,B,\omega,t,T)$ of \probBORDF.

We use the algorithm from Lemma~\ref{lem:unbreakable} for $G$. This algorithm in time $2^{2^{\Oh(t\log t)}}\cdot n^{\Oh(1)}$ either  finds a separation $(U,W)$ of $G$ of order at most $t$ such that $|U\setminus W|>q$ and $|W\setminus U|>q$ or correctly reports that $G$ is $((2q+1)q\cdot 2^t,t)$-unbreakable. In the latter case we solve the problem using Lemma~\ref{lem:unbreak} in time $2^{2^{2^{\Oh(t\log t)}}}\cdot n^{\Oh(1)}$. 
Assume from now that there is a separation  $(U,W)$ of order at most $t$ such that $|U\setminus W|>q$ and $|W\setminus U|>q$.

Recall that $|T|\leq 2t$. Then $|T\cap (U\setminus W)|\leq t$ or $|T\cap (W\setminus U)|\leq t$. Assume without loss of generality that $|T\cap (W\setminus U)|\leq t$. Let $\tilde{G}=G[W]$, $\tilde{I}=I\cap W$, $\tilde{O}=O\cap W$, $\tilde{\omega}$ is the restriction of $\omega$ to $W$, and define $\tilde{T}=(T\cap W)\cup (U\cap W)$. Since $|U\cap W|\leq t$, $|\tilde{T}|\leq 2t$.

If $|W|\leq (2q+1)q\cdot 2^t$, then we solve \probBORDF for the instance $(\tilde{G},\tilde{I},\tilde{O},\tilde{B},\tilde{\omega},t,\tilde{T})$ by brute force in time $2^{2^{2^{\Oh(t\log t)}}}$ trying all possible subset of $W$ at most $t+1$ values of $0\leq t'\leq t$. Otherwise, we solve  $(\tilde{G},\tilde{I},\tilde{O},\tilde{B},\tilde{\omega},t,\tilde{T})$ recursively. Let $\mathcal{R}$ be the set of nonempty induced subgraphs $R$ that are included in the obtained solution for $(\tilde{G},\tilde{I},\tilde{O},\tilde{B},\tilde{\omega},t,\tilde{T})$. 

For $R\in \mathcal{R}$, define $S_R$ to be the set of vertices of $W\setminus V(R)$ that are adjacent to the vertices of $R$ in the graph obtained by the border complementation for which $R$ is a solution of the corresponding instance of \probMAXF. Note that $|S_R|\leq t$.
If $\mathcal{R}\neq\emptyset$, then let $S=\tilde{T}\cup_{R\in\mathcal{R}}S_R$, and $S=\tilde{T}$ if $\mathcal{R}=\emptyset$.
Since \probMAXF is solved for at most $t+1$ of values of $t'\leq t$, at most $3^{2t}$ three-partitions $(X,Y,Z)$ of $\tilde{T}$ and at most $2^{c2t\log (2t)}$ choices of a $p$-boundaried graph $H\in \mathcal{F}_b$ for $p=|X|$, we have that $|\mathcal{R}|\leq (t+1)3^{2t}2^{c2t\log (2t)}$. Taking into account that $|T'|\leq 2t$, 
\begin{equation}\label{Aeq:s}
|S|\leq (t+1)t3^{2t}2^{c2t\log (2t)}+2t.
\end{equation}

Let $\hat{B}=(B\cap U)\cup (B\cap S)$. We claim that the instances $(G,I,O,B,\omega,t,T)$ and $(G,I,O,\hat{B},\omega,t,T)$ of \probBORDF are equivalent. 

Recall that we have to show that 
\begin{itemize}
\item[(i)]  $T\cap B=T\cap \hat{B}$,
\item[(ii)] for the border complementations $(G',I',O',B',\omega',t')$ and $(G',I',O',\hat{B}',\omega',t')$ of the instances $(G,I,O,B,\omega,t')$ and $(G,I,O,\hat{B},\omega,t')$ respectively of \probMAXF with respect to every feasible $(X=\{x_1,\ldots,x_p\}$, $Y,Z,H)$ and $t'\leq t$, it holds that if $(G',I',O',B',\omega',t')$ has a nonempty solution $R_1$, then $(G',I',O',\hat{B}',\omega',t')$ has a nonempty solution $R_2$ with $\omega'(V(R_2))\geq \min\{\omega'(V(R_1)),w\}$ and, vice versa, if $(G',I',O',\hat{B}',\omega',t')$ has a nonempty solution $R_2$, then $(G',I',O',B',\omega',t')$ has a nonempty solution $R_1$ with $\omega'(V(R_1))\geq \min\{\omega'(V(R_2)),w\}$.
\end{itemize}

The condition (i) holds by the definition of $\hat{B}$. Because $\hat{B}\subseteq B$,  we immediately obtain that if $(G',I',O',\hat{B}',\omega',t')$ has a nonempty solution $R_2$, then $(G',I',O',B',\omega',t')$ has a nonempty solution $R_1$ with $\omega'(V(R_1))\geq \min\{\omega'(V(R_2)),w\}$. 
It remains to prove that  for a border complementation $(G',I',O',B',\omega',t')$ and $(G',I',O',\hat{B}',\omega',t')$ of $(G,I,O,B,\omega,t')$ and $(G,I,O,\hat{B},\omega,t')$ respectively of \probMAXF with respect to a feasible $(X=\{x_1,\ldots,x_p\},Y,Z,H)$ and $t'\leq t$, it holds that if $(G',I',O',B',\omega',t')$ has a nonempty solution $R_1$, then $(G',I',O',\hat{B}',\omega',t')$ has a nonempty solution $R_2$ with $\omega'(V(R_2))\geq \min\{\omega'(V(R_1)),w\}$.

If $V(R_1)\cap V(G)\subseteq U\setminus W$, then $N_{G'}(V(R_1))\subseteq \hat{B}'$. Therefore, for a solution $R_2$ of $(G',I',O',\hat{B}',\omega',t')$, $\omega'(V(R_2))\geq \min\{\omega'(V(R_1)),w\}$. Assume that $V(R_1)\cap W\neq \emptyset$. 
Let $\tilde{X}=\tilde{T}\cap(V(R_1)\cap W)=\{y_1,\ldots,y_r\}$, let $\tilde{Y}$ be the set of vertices of $\tilde{T}\setminus V(R_1)$ that are adjacent to vertices of $R_1$
outside  $W\setminus U$ and $\tilde{Z}=\tilde{T}\setminus (\tilde{X}\cup \tilde{Y})$. 
Let $(R_1',(y_1,\ldots,y_r))$ be the $r$-bounderied graph obtained from $R_1$ by the deletion of the vertices of $W\setminus U$ (note that the graph could be empty).
We have that $\mathcal{G}_r$ contains an $r$-boundaried graph $\tilde{H}$ that is equivalent to $(R_1',(y_1,\ldots,y_r))$. 
Recall that we have a solution of \probBORDF for $(\tilde{G},\tilde{I},\tilde{O},\tilde{B},\tilde{\omega},t,\tilde{T})$. In particular, we have a solution $\tilde{R}\in \mathcal{R}$ for the instance 
$(\tilde{G},\tilde{I},\tilde{O},\tilde{B},\tilde{\omega},\tilde{t})$ of \probMAXF obtained by the border complementation with respect to $(\tilde{X}=(y_1,\ldots,y_r),\tilde{Y},\tilde{Z},\tilde{H})$, where $\tilde{t}$ is the number of neighbors of $R_1$ in $W$.  
Recall also that the neighbors of the vertices of $\tilde{R}$ are in $S$.
Denote by $(\tilde{R}',(y_1,\ldots,y_r))$ the $r$-bounderied subgraph obtained from $\tilde{R}$ by the deletion of the vertices that are outside of  $W$. 
By Lemma~\ref{lem:eq}, $R_2=(R_1',(y_1,\ldots,y_r))\oplus_b (\tilde{R},(y_1,\ldots,y_r))$ is $\mathcal{F}$-free. Observe also that $\omega'(R_2)\geq \min\{\omega'(R_1),w\}$. It implies that 
 $(G',I',O',\hat{B}',\omega',t')$ has a nonempty solution $R_2$ with $\omega'(V(R_2))\geq \min\{\omega'(V(R_1)),w\}$.

Since, $(G,I,O,B,\omega,t,T)$ and $(G,I,O,\hat{B},\omega,t,T)$ of \probBORDF are equivalent, we can consider $(G,I,O,\hat{B},\omega,t,T)$. Now we apply some reduction rules that produce equivalent instances of \probBORDF or report that we have no solution. The ultimate aim of these rules is to reduce the size of $G$. 

Let $Q$ be a component of $G[W]-S$. Notice that for any nonempty graph  $R$ in a solution of $(G,I,O,\hat{B},w,t,T)$, either $V(Q)\subseteq V(R)$ or $V(Q)\cap V(R)=\emptyset$, because $N_{G[W]}(V(R))\subseteq S$. Moreover, if
$V(Q)\cap V(R)=\emptyset$, then $N_{G[W]}[V(Q)]\cap V(R)=\emptyset$. 
Notice also that if $v\in N_{G[W]}(V(Q))$ is a vertex of $R$, then  $V(Q)\subseteq V(R)$.
These observation are crucial for the following reduction rules.

\begin{reduction}\label{Ared:one-rec}
For a component $Q$ of $G[W]-S$ do the following in the given order:
\begin{itemize}
\item if $N_{G[W]}[V(Q)]\cap I\neq\emptyset$ and $V(Q)\cap O\neq\emptyset$, then return $\emptyset$ and stop,
\item if $N_{G[W]}[V(Q)]\cap I\neq\emptyset$, then set $I=I\cup V(Q)$,
\item if $V(Q)\cap O\neq\emptyset$, then set $O=O\cup N_{G[W]}[V(Q)]$.
\end{itemize}
\end{reduction}

The rule is applied to each component $Q$ exactly once. Notice that after application of the rule, for every component $Q$ of $G[W]-S$, we have that either $V(Q)\subseteq I$ or $V(Q)\subseteq O$ or $V(Q)\cap (I\cup O\cup\hat{B})=\emptyset$.

Suppose that $Q_1$ and $Q_2$ are components of $G[W]-S$ such that $N_{G[W]}(V(Q_1))=N_{G[W]}(V(Q_2))$ and $|N_{G[W]}(V(Q_1))|=|N_{G[W]}(V(Q_2))|>t$. Then if $V(Q_1)\subseteq V(R)$ for a 
nonempty graph  $R$ in a solution of $(G,I,O,\hat{B},\omega,t,T)$, then at least one vertex of $N_{G[W]}(V(Q_1))$ is in $R$ as  $R$ have at most $t$ neighbors outside $R$. This gives the next rule.

\begin{reduction}\label{Ared:two-rec}
For components $Q_1$ and $Q_2$ of $G[W]-S$ such that  $N_{G[W]}(V(Q_1))=N_{G[W]}(V(Q_2))$ and $|N_{G[W]}(V(Q_1))|=|N_{G[W]}(V(Q_2))|>t$ do the following in the given order:
\begin{itemize}
\item if $(V(Q_1)\cup V(Q_2))\cap I\neq\emptyset$ and $(V(Q_1)\cup V(Q_2))\cap O\neq\emptyset$, then return $\emptyset$ and stop,
\item if $(V(Q_1)\cup V(Q_2))\cap I\neq\emptyset$, then set $I=I\cup (V(Q_1)\cup V(Q_2))$,
\item if $(V(Q_1)\cup V(Q_2))\cap O\neq\emptyset$, then set $O=O\cup N_{G[W]}[V(Q_1)\cup V(Q_2)]$.
\end{itemize}
\end{reduction}

We apply the rule for all pairs of components $Q_1$ and $Q_2$ with $N_{G[W]}(V(Q_1))=N_{G[W]}(V(Q_2))$ and $|N_{G[W]}(V(Q_1))|=|N_{G[W]}(V(Q_2))|>t$, and for each pair the rule is applied once. 

If $V(Q)\subseteq O$ for a component $Q$ of $G[W]-S$, then $N_{G[W]}(V(Q))\subseteq O$. It immediately implies that the vertices of $Q$ are irrelevant and can be removed.

\begin{reduction}\label{Ared:del-rec}
If there is a component $Q$ of $G[W]-S$ such that  $N_{G[W]}(V(Q))\subseteq O$, then 
set $G=G-V(Q)$, $W=W\setminus V(Q)$ and $O=O\setminus V(Q)$.
\end{reduction}

Notice that for each component $Q$, we have now that either $V(Q)\subseteq I$ or $V(Q)\subseteq W\setminus (I\cup O\cup\hat{B})$.

To define the remaining rules, we construct
 the sets $\mathcal{G}_p'$ for $p\in \{0,\ldots,|S|\}$ in time $2^{2^{\Oh(t\log t)}}$ using Lemma~\ref{lem:eq-constr}. 

Let $Q$ be a component of $G[W]-S$ and let $N_{G[W]}(V(Q))=\{x_1,\ldots,x_p\}$. Let $G'$ be the graph obtained from $G$ by the deletion of the vertices of $V(Q)$ and let $x=(x_1,\ldots,x_p)$. Let $(H,y)$ be a connected $p$-boundaried graph of the same weight as $G[N_{G[W]}[V(Q)]]$.  Then by Lemma~\ref{lem:eq}, we have that the instance of \probBORDF obtained from 
$(G,I,O,\hat{B},\omega,t,T)$ by the replacement of $G$ by $(G',x)\oplus_b (H,y)$ is equivalent to $(G,I,O,\hat{B},\omega,t,T)$. We use it in the remaining reduction rules.

Suppose again that $Q_1$ and $Q_2$ are components of $G[W]-S$ such that $N_{G[W]}(V(Q_1))=N_{G[W]}(V(Q_2))$ and $|N_{G[W]}(V(Q_1))|=|N_{G[W]}(V(Q_2))|>t$. 
Then, as we already noticed, if $V(Q_1)\cup V(Q_2)\subseteq V(R)$ for a 
nonempty graph  $R$ in a solution of $(G,I,O,\hat{B},\omega,t,T)$, then at least one vertex of $N_{G[W]}(V(Q_1))$ is in $R$. 
It means that if we are constructing a solution $R$, then the restriction of the size of the neighborhood of $R$ ensures the connectivity between $Q_1$ and $Q_2$ if we decide to include these components in $R$. 
Together with Lemma~\ref{lem:eq} this shows that the following rule is safe.

\begin{reduction}\label{Ared:three-rec}
Let $L=\{x_1,\ldots,x_p\}\subseteq S$, $p>t$, and let $x=(x_1,\ldots,x_p)$. Let also $Q_1,\ldots,Q_r$,$r\geq 1$, be the components of $G[W]-S$ with $N_{G[W]}(V(Q_i))=L$ for all $i\in\{1,\ldots,r\}$.
Let $Q=G[\cup_{i=1}^rN_{G[W]}[V(Q_i)]]$ and $w'=\sum_{i=1}^r\omega(V(Q_i))$. Find a $p$-boundaried graph $(H,y)\in\mathcal{G}_p'$ 
 that is equivalent to $(Q,x)$ with respect to $\mathcal{F}_b$ and denote by $A$ the set of nonboundary vertices of $H$.
Then do the following.
\begin{itemize}
\item Delete the vertices of $V(Q_1),\ldots,V(Q_r)$ from $G$ and denote the obtained graph $G'$.
\item Set $G=(G',x)\oplus_b (H,y)$ and $W=(W\setminus\cup_{i=1}^rV(Q_i))\cup A$.
\item Select arbitrarily $u\in A$ and modify $\omega$ as follows:
\begin{itemize}
\item keep the weight same for every $v\in V(G')$ including the boundary vertices $x_1,\ldots,x_p$,
\item set $\omega(v)=0$ for $v\in A\setminus\{u\}$,
\item set $\omega(u)=w'$.
\end{itemize} 
\item If $V(Q_1)\subseteq I$, then set $I=I\setminus(\cup_{i=1}^rV(Q_i))\cup A$.
\end{itemize}
\end{reduction}

To see the safeness of the rule, observe additionally that the neighborhood of  each component of $H[A]$ is $L$, because $(H,y)\in \mathcal{G}_p'$. 
The rule is applied exactly once for each inclusion maximal sets of components $\{Q_1,\ldots,Q_r\}$ having the same neighborhood of size at least $t+1$.

We cannot apply this trick if we have several components $Q_1,\ldots,Q_r$ of $G[W]-S$ with the same neighborhood $N_{G[W]}V(Q_i)$ if $|N_{G[W]}V(Q_i)|\leq t$. Now it can happen that there are $i,j\in \{1,\ldots,r\}$ such that $V(Q_i)\subseteq V(R)$ and $N_{G[W]}[V(Q_j)]\cap V(R)=\emptyset$ for  $R$ in a solution of $(G,I,O,\hat{B},\omega,t,T)$. But if $N_{G[W]}[V(Q_j)]\cap V(R)=\emptyset$ , then by the connectivity of $R$ and the fact that $G[W]-S$ does not contain border terminals, we have that $R=Q_i$. Notice that $I=\emptyset$ in this case and, in particular, it means that $R$ is a solution for an instance of \probMAXF obtained by the border complementation with respect to $(\emptyset,\emptyset,T,\emptyset)$. Recall that we output $R$ in this case only if its weight is at least $w$. 

Still, we can modify Reduction Rule~\ref{Ared:three-rec} for the case when there are components $Q$ of $G[W]-S$ such that $V(Q)\subseteq I$. Notice that if there are components 
 $Q_0,\ldots,Q_r$ of $G[W]-S$ with the same neighborhood and $V(Q_0)\subseteq I$, then for any nonempty $R$ in a solution of  $(G,I,O,\hat{B},\omega,t,T)$, either $R=Q_0$ or $\cup_{i=0}^rV(Q_i)\subseteq V(R)$. Applying Lemma~\ref{lem:eq} , we obtain that the following rule is safe.

\begin{reduction}\label{Ared:four-rec}
Let $L=\{x_1,\ldots,x_p\}\subseteq S$, $p\leq t$, and let $x=(x_1,\ldots,x_p)$. Let also $Q_0,\ldots,Q_r$,$r\geq 0$, be the components of $G[W]-S$ with $N_{G[W]}(V(Q_i))=L$ for all $i\in\{0,\ldots,r\}$ such that $V(Q_0)\subseteq I$. Let $Q=G[\cup_{i=1}^rN_{G[W]}[V(Q_i)]]$ and $w'=\sum_{i=1}^r\omega(V(Q_i))$.
Find a $p$-boundaried graph $(H_0,y)\in\mathcal{G}_p'$ 
 that is equivalent to $(Q_0,x)$ with respect to $\mathcal{F}_b$  and denote by $A_0$ the set of nonboundary vertices of $H_0$, and 
find a  $p$-boundaried graph $(H,y)\in\mathcal{G}_p'$ 
 that is equivalent to $(Q,x)$ with respect to $\mathcal{F}_b$ and denote by $A$ the set of nonboundary vertices of $H$.
Then do the following.
\begin{itemize}
\item Delete the vertices of $V(Q_0),\ldots,V(Q_r)$ from $G$ and denote the obtained graph $G'$.
\item Set $G=(((G',x)\oplus_b (H_0,y)),y)\oplus_b (H,y)$ and $W=(W\setminus\cup_{i=0}^rV(Q_i))\cup A_0\cup A$.
\item Select arbitrarily $u\in A_0$ and $v\in A$ and modify $\omega$ as follows:
\begin{itemize}
\item keep the weight same for every $z\in V(G')$ including the boundary vertices $x_1,\ldots,x_p$,
\item set $\omega(z)=0$ for $z\in (A_0\setminus\{u\})\cup(A\setminus\{v\})$,
\item set $\omega(u)=\omega(V(Q_0))$ and $\omega(v)=w'$.
\end{itemize} 
\item If $V(Q_i)\subseteq I$ for some $i\in\{1,\ldots,r\}$, then set $I=I\setminus(\cup_{i=1}^rV(Q_i))\cup A$.
\end{itemize}
\end{reduction}

To see the safeness, notice additionally that $H_0[A_0]$ is connected, because $(H_0,y)\in \mathcal{G}_p'$ and $Q_0$ is connected.
The rule is applied exactly once for each inclusion maximal set of components $\{Q_1,\ldots$, $Q_r\}$ having the same neighborhood of size at most $t$ such that $V(Q_i)\subseteq I$ for some $i\in\{1,\ldots,r\}$.

Assume now that we have an inclusion maximal set of components $\{Q_1,\ldots$ ,$Q_r\}$ of $G[W]-S$ with the same neighborhoods $N_{G[W]}=\{x_1,\ldots,x_p\}$ such that $(G[N_{G[W]}[V(Q_i)]],(x_1,\ldots,x_p))$ and $(G[N_{G[W]}[V(Q_j)]],(x_1,\ldots,x_p))$ are equivalent with respect to $\mathcal{F}_b$ for each $i,j\in\{1,\ldots,p\}$. Suppose also that $V(Q_i)\cap I=\emptyset$ for $i\in\{1,\ldots,r\}$.
 Let $\omega(V(Q_1))\geq  \omega(V(Q_i))$ for every $i\in\{1,\ldots,r\}$. Recall that if $R$ is a nonempty graph in a solution, then either $R=Q_i$ for some $i\in \{1,\ldots,r\}$ or $\cup_{i=1}^r V(Q_i)\subseteq V(R)$. Recall also that   $R$ is a solution for the instance of \probMAXF obtained by a border complementation with respect to $(\emptyset,\emptyset,T,\emptyset)$ and we output it only if $\omega(V(R))\geq w$. Since 
all $(G[N_{G[W]}[V(Q_i)]],(x_1,\ldots,x_p))$ are equivalent, we can assume that if $R=Q_i$, then $i=1$, because $Q_1$ has maximum weight. Then by Lemma~\ref{lem:eq}, our final reduction rule is safe.

\begin{reduction}\label{Ared:five-rec}
Let $L=\{x_1,\ldots,x_p\}\subseteq S$, $p\leq t$, and let $x=(x_1,\ldots,x_p)$. Let also $Q_0,\ldots,Q_r$,$r\geq 0$, be the components of $G[W]-S$ with $N_{G[W]}(V(Q_i))=L$ for all $i\in\{0,\ldots,r\}$ such that $\omega(V(Q_0))\geq \omega(V(Q_i))$ for every $i\in\{1,\ldots,r\}$ and 
the $p$-boundaried graphs $(G[N_{G[W]}[V(Q_i)]],(x_1,\ldots,x_p))$ are pairwise equivalent with respect to $\mathcal{F}_b$ for $i\in\{0,\ldots,r\}$. 
Let $Q=G[\cup_{i=1}^rN_{G[W]}[V(Q_i)]]$ and $w'=\min\{w-1,\sum_{i=1}^r\omega(V(Q_i))\}$.
Find a $p$-boundaried graph $(H_0,y)\in\mathcal{G}_p'$ 
 that is equivalent to $(Q_0,x)$ with respect to $\mathcal{F}_b$  and denote by $A_0$ the set of nonboundary vertices of $H_0$, and 
find a  $p$-boundaried graph $(H,y)\in\mathcal{G}_p'$ 
 that is equivalent to $(Q,x)$ with respect to $\mathcal{F}_b$ and denote by $A$ the set of nonboundary vertices of $H$.
Then do the following.
\begin{itemize}
\item Delete the vertices of $V(Q_0),\ldots,V(Q_r)$ from $G$ and denote the obtained graph $G'$.
\item Set $G=(((G',x)\oplus_b (H_0,y)),y)\oplus_b (H,y)$ and $W=(W\setminus\cup_{i=0}^rV(Q_i))\cup A_0\cup A$.
\item Select arbitrarily $u\in A_0$ and $v\in A$ and modify $\omega$ as follows:
\begin{itemize}
\item keep the weight same for every $z\in V(G')$ including the boundary vertices $x_1,\ldots,x_p$,
\item set $\omega(z)=0$ for $z\in (A_0\setminus\{u\})\cup(A\setminus\{v\})$,
\item set $\omega(u)=\omega(V(Q_0))$ and $\omega(v)=w'$.
\end{itemize} 
\end{itemize}
\end{reduction}

Notice that we upper bound the weight of $A$ by $w-1$ to prevent selecting $R=G[A]$ as a graph in a solution. To see that it is safe, observe that if $\omega(V(Q_0))>0$, then the total weight of $A_0$ and $A$ is at least $w$ and recall that by the definition of \probMAXF, we output nonempty graphs of maximum weight or weight at least $w$. Therefore, if we output $R$ that includes $A_0\cup A$, then we output a graph of weight at least $w$. If $\omega(Q_0)=0$, then $w'=\sum_{i=1}^r\omega(V(Q_i))=0$. 
  
The Reduction Rule~\ref{Ared:five-rec} is applied for each inclusion maximal sets of components $\{Q_1,\ldots$, $Q_r\}$ satisfying the conditions of the rule such that Reduction Rule~\ref{Ared:four-rec} was not applied to these components before.

Denote by $(G^*,I^*,O^*,B^*,\omega^*,t,T)$ the instance of \probBORDF obtained from  $(G,I,O,\hat{B},\omega,t,T)$ by Reduction Rules~\ref{Ared:one-rec}-\ref{Ared:five-rec}. Notice that all modifications were made for $G[W]$.
Denote by $W^*$ the set of vertices of the graph obtained from the initial $G[W]$ by the rules. Observe that there are at most $2^{|S|}$ subsets $L$ of $S$ such that there is a component $Q$ of $G[W]-S$ with $N_{G[W]}(V(Q))=L$. If $|L|>t$, then all $Q$ with $N_{G[W]}(V(Q))=L$ are replaced by one graph  by Reduction Rule~\ref{Ared:three-rec} and the number of vertices of this graph is at most $|L|^c$ by Lemma~\ref{lem:size-eq} and the definition of $c$. If $|L|\leq t$, then we either apply Reduction Rule~\ref{Ared:four-rec} for all $Q$ with $N_{G[W]}(V(Q))=L$ and replace these components by two graph with at most  $|L|^c$ vertices or we apply Reduction Rule~\ref{Ared:five-rec}. For the latter case, observe that there are at most $t^c$ partitions of the components $Q$ with $N_{G[W]}(V(Q))=L$ into  equivalence classes with respect to $\mathcal{F}_b$
by Lemma~\ref{lem:size-eq}. Then we replace each class by two graphs with at most $|L|^c$ vertices.
Taking into account the vertices of $S$, we obtain the following upper bound for the size of $W^*$:
\begin{equation}\label{Aeq:w*}
|W^*|\leq 2^{|S|}2|S|^c t^c+|S|.
\end{equation}
By (\ref{Aeq:q}) and (\ref{Aeq:s}), $|W^*|\leq q$.
Recall that $|W\setminus U|>q$. Therefore, $|V(G^*)|<|V(G)|$. We use it and solve \probBORDF for $(G^*,I^*,O^*,B^*,\omega^*,t,T)$ recursively.

To evaluate the running time, denote by $\tau(G,I,O,B,\omega,t,T)$ the time needed to solve \linebreak \probBORDF for \linebreak $(G,I,O,B,\omega,t,T)$.
Lemmas~\ref{lem:eq-check} and \ref{lem:size-eq} imply that the reduction rules are polynomial. The algorithm from Lemma~\ref{lem:unbreakable} runs in time $2^{2^{\Oh(t\log t)}}\cdot n^{\Oh(1)}$. Notice that the sets $\mathcal{G}_p$ and $\mathcal{G}_p'$ can be constructed separately from the algorithm for \probBORDF. Then we obtain the following recurrence for the running time:
\begin{equation}\label{Aeq:run}
\tau(G,I,O,B,\omega,t,T)\leq \tau(G^*,I^*,O^*,B^*,\omega^*,t,T)+\tau(\tilde{G},\tilde{I},\tilde{O},\tilde{B},\tilde{\omega},t,\tilde{T})+2^{2^{\Oh(t\log t)}}\cdot n^{\Oh(1)}.
\end{equation}
Because $|W^*|\leq q$,
\begin{equation}\label{Aeq:vert}
|V(G^*)|\leq |V(G)|-|V(\tilde{G})|+q.
\end{equation}
Recall that if the algorithm of Lemma~\ref{lem:unbreakable} reports that $G$ is $((2q+1)q\cdot 2^t,t)$-unbreakable or we have that $|V(G)|\leq ((2q+1)q\cdot 2^t$, we do not recurse but solve the problem directly in time $2^{2^{2^{\Oh(t\log t)}}}\cdot n^{\Oh(1)}$. Following the general scheme from~\cite{ChitnisCHPP16}, we obtain that these condition together with (\ref{Aeq:run}) and (\ref{Aeq:vert}) imply that the total running time is  $2^{2^{2^{\Oh(t\log t)}}}\cdot n^{\Oh(1)}$.
\end{proof}

Now have now all the details in place to be able to prove Theorem \ref{thm:forb}, stating that \probLForb is \classFPT when parameterized by $t$.

\medskip
\noindent
{\bf Theorem~\ref{thm:forb}.}{\it
\probLForb can be solved in time  $2^{2^{2^{\Oh(t\log t)}}}\cdot n^{\Oh(1)}$.
}

\begin{proof}
Let $(G,\omega,t,w)$ be an instance of \probLForb. We define $I=\emptyset$, $O=\emptyset$, $B=V(G)$ and $T=\emptyset$. Then we solve \probBORDF for $(G,I,O,B,w,t,T)$ using Lemma~\ref{lem:bordforb}
 in time $2^{2^{2^{\Oh(t\log t)}}}\cdot n^{\Oh(1)}$. It remains to notice that $(G,\omega,t,w)$ is a yes-instance of \probLForb if and only if $(G,I,O,B,\omega,t,T)$ has a nonempty graph in a solution.
\end{proof}

\section{Large Secluded Trees}\label{sec:fpt-tree}
In this section we show that \probLCSS is \classFPT when parameterized by $t$ when $\Pi$ is defined by a infinite set of forbidden induced subgraphs, namely, by the set of cycles.
In other words, a graph $G$ has the property $\Pi$ considered in this section if $G$ is a forest. 
We refer to this problem as \probLT. We again apply the recursive understanding technique introduced by Chitnis et al.~\cite{ChitnisCHPP16} and follow the scheme of the previous section.
In particular, we solve a special variant of \probLT tailored for recursion.

We define the following auxiliary problem for a positive integer $w$.  

 \defproblema{\probMAXT}%
{A graph $G$, sets $I,O,B\subseteq V(G)$ such that $I\cap O=\emptyset$ and $I\cap B=\emptyset$, a weight function $\omega\colon V(G)\rightarrow\mathbb{Z}_{\geq 0}$ and a nonnegative integer $t$.}%
{Find a $t$-secluded  induced connected subtree $H$ of $G$ of maximum weight or weight at least $w$
such that $I\subseteq V(H)$, $O\subseteq V(G)\setminus V(H)$ and $N_G(V(H))\subseteq B$ and output $\emptyset$ if such a subgraph does not exist.}

Let $(G,I,O,B,\omega,t)$ be an instance of \probMAXT and let $T\subseteq V(G)$ be a set of \emph{border terminals}.
We say that a 4-tuple $(X,Y,Z,\mathcal{P})$, where  $(X,Y,Z)$ is a partition of $T$ (some sets could be empty) and $\mathcal{P}=(P_1,\ldots,P_s)$ is a partition of $X$ into nonempty sets if $X\neq\emptyset$ and $\mathcal{P}=\emptyset$ otherwise, is  \emph{feasible} if $G[X]$ is a forest,
$Y=\emptyset$ if $X=\emptyset$, $I\cap T\subseteq X$, $O\cap T\subseteq Y\cup Z$ and $Y\subseteq B$, 
and if $X\neq\emptyset$, then the vertices of each component of $G[X]$ are in the same set of the partition $\mathcal{P}$.
 We say that an instance $(G',I',O',B',\omega',t')$ is obtained by the \emph{border complementation (with respect to feasible 4-tuple $(X,Y,Z,\mathcal{P})$ )} for $X\neq\emptyset$ and  $\mathcal{P}=(P_1,\ldots,P_s)$ if 
\begin{itemize} 
\item[(i)] $G'$ is obtained from $G$ by adding vertices $u_1,\ldots,u_s$ and making $u_i$ adjacent to an arbitrary vertex of each component of $G[X]$ in $G[P_i]$ 
and to the vertices of $Y$ for $i\in \{1,\ldots,s\}$,
\item[(ii)] $I'= I\cup\{u\}$,
\item[(iii)] $O'=O\cup Y\cup Z$,
\item[(iv)] $B'=B\setminus X$,
\item[(v)] $\omega'(v)=\omega(v)$ for $v\in V(G)$ and $\omega'(u)=0$,
\item[(vi)] $t'\leq t$.
\end{itemize}
If $X=\emptyset$, then $(G',I',O',B',w',t')$ is obtained by the \emph{border complementation (with respect to $(X,Y,Z,\mathcal{P})$ )} if 
$G'=G$, $I'=I$, $O'=O\cup T$, $B'=B$, $\omega'(v)=\omega(v)$ for $v\in V(G)$ and $t'\leq t$.

 \defproblema{\probBORDT}%
{A graph $G$, sets $I,O,B\subseteq V(G)$ such that $I\cap O=\emptyset$ and $I\cap B=\emptyset$, a weight function $\omega\colon V(G)\rightarrow\mathbb{Z}_{\geq0}$, a nonnegative integer $t$, and a set $T\subseteq V(G)$ of border terminals of size at most $2(t+1)$.}{Output a solution for each instance $(G',I',O',B',\omega',t')$ of \probMAXT that can be obtained from $(G,I,O,B,w,t)$ by a border complementation distinct from the border complementation with respect to $(\emptyset,\emptyset,T,\emptyset)$, and for the border complementation with respect to $(\emptyset,\emptyset,T,\emptyset)$ output a nonempty solution if it has weight at least $w$ and output $\emptyset$ otherwise.}

Two instances $(G_1,I_1,O_1,B_1,\omega_1,t,T)$ and $(G_2,I_2,O_2,B_2,\omega_2,t,T)$ of \probBORDT (note that $t$ and $T$ are the same) are said to be \emph{equivalent} if 
\begin{itemize}
\item[(i)] $T\cap I_1=T\cap I_2$, $T\cap O_1=T\cap O_2$ and $T\cap B_1=T\cap B_2$,
\item[(ii)] for the border complementations $(G_1',I_1',O_1',B_1',\omega_1',t')$ and $(G_2',I_2',O_2',B_2',\omega_2',t')$ of the instances $(G_1,I_1,O_1,B_1,\omega_1,t')$ and $(G_2,I_2,O_2,B_2,\omega_2,t')$ respectively of \probMAXF with respect to every feasible $(X,Y,Z,\mathcal{P})$ and $t'\leq t$, it holds that if $(G_1',I_1',O_1',B_1',\omega_1',t')$ has a nonempty solution $R_1$, then $(G_2',I_2',O_2',B_2',\omega_2',t')$ has a nonempty solution $R_2$ with $\omega_2'(V(R_2))\geq \min\{\omega_1'(V(R_1)),w\}$ and, vice versa, if $(G_2',I_2',O_2',B_2',\omega_2',t')$ has a nonempty solution $R_2$, then $(G_1',I_1',O_1',B_1',\omega_1',t')$ has a nonempty solution $R_1$ with $\omega_1'(V(R_1))\geq \min\{\omega_2'(V(R_2))$, $w\}$.
\end{itemize}
As in the previous section we not distinguish equivalent instances of \probBORDT and their solutions.

\subsection{High connectivity phase}
In this section we solve \probBORDT for $(q,t+1)$-unbreakable graphs. We need the following folklore lemma.

\begin{lemma}\label{lem:balance}
Every tree $T$ has a separation $(A,B)$ of order 1 such that $|A\setminus B|\leq \frac{2}{3}|V(T)|$ and $|B\setminus A|\leq \frac{2}{3}|V(T)|$.
\end{lemma}

\begin{lemma}\label{lem:unbreak-tree}
\probBORDT for $(q,t+1)$-unbreakable graphs can be solved in time  $2^{\Oh((t+\min\{q,t\})\log(q+t))}\cdot n^{\Oh(1)}$.
\end{lemma}

\begin{proof}
Consider an instance $(G',I',O',B',\omega',t')$ of \probMAXT be obtained from  $(G,I,O,B,\omega,t)$  by the border complementation with respect to some feasible $(X,Y,Z,\mathcal{P})$. 
Assume that $H$ is a nonempty solution of $(G',I',O',B',\omega',t')$. We claim that $|V(H)|\leq 3q+8$.

To obtain a contradiction, assume that $|V(H)|\geq 3q+9$. By Lemma~\ref{lem:balance}, there is a separation $(U,W)$ of $H$ of order 1 such that $|U\setminus W|\leq \frac{2}{3}|V(H)|$ and $|W\setminus U|\leq \frac{2}{3}|V(H)|$, that is, $|U\setminus W|\geq \frac{1}{3}|V(H)|-1\geq q+2$ and $|W\setminus U|\geq \frac{1}{3}|V(H)|-1\geq q+2$. 
Let $U'=U\cap V(G)$ and $W'=W\cap V(G)$. We have that $|U'\setminus W'|\geq q+1$ and $|W'\setminus U'|\geq q+1$.
Let $U''=U'\cup N_{G'}(V(H))\subseteq V(G)$ and $W''=V(G)\setminus (U'\setminus W')\subseteq V(G)$.
We have that $(U'',W'')$ is a separation of $G$.
Since $U''\cap W''=N_{G'}(V(H))\cup (U'\cap W')$, we obtain that the order of the separation is at most $t+1$.
Observe that $U''\setminus W''=U'\setminus W''$ and 
$W'\setminus U'\subseteq W''\setminus U''$ and, therefore, $|U''\setminus W''|\geq q+1$ and $|W''\setminus U''|\geq q+1$ contradicting the $(q,t+1)$-unbreakability of $G$.

The claim implies that to solve \probMAXT for \linebreak $(G',I',O',B',\omega',t')$, it is sufficient to consider $k\leq 3q+8$ and for each $k$,  find $t$-secluded induced subtree $H$ in $G'$ of maximum weight 
such that 
$I'\subseteq V(H')$, $O'\subseteq V(G')\setminus V(H)$, $N_{G'}(V(H))\subseteq B'$ and $|V(H)|=k$. By Corollary~\ref{cor:CSWCS},
it can be done in time $2^{\Oh(\min\{q,t\}\log (q+t))}\cdot n^{\Oh(1)}$. 

To solve \probBORDF for $(G,I,O,B,\omega,t,T)$, we consider all possible partitions $(X,Y,Z)$ of $T$ and partitions $\mathcal{P}$ of $X$ such that $(X,Y,Z,\mathcal{P})$ is feasible. 
Since $|T|\leq 2(t+1)$, and the set of $\mathcal{P}$ together with $X$ and $Y$ form a partition of $T$, we conclude that 
there are $2^{\Oh(t\log t)}$ feasible $(X,Y,Z,\mathcal{P})$. Hence, the total running time is $2^{\Oh((t+\min\{q,t\})\log (q+t))}\cdot n^{\Oh(1)}$.
\end{proof}

\subsection{The \classFPT algorithm for \probLT}\label{Asec:tree}
In this section we construct an \classFPT algorithm for \probLT parameterized by $t$. To do it, we solve \probBORDF in \classFPT-time for general case.

\begin{lemma}\label{lem:bordtree}
\probBORDT  can be solved in time $2^{2^{\Oh(t\log t)}}\cdot n^{\Oh(1)}$.
\end{lemma}

\begin{proof}
There is a constant $c$ such that the number of partitions of a $k$-element set into subsets such that at most two of them could be empty is at most $2^{ck\log k}$.
We define 
\begin{equation}\label{eq:q-t}
q=2\cdot 2^{(t+1)t2^{c2(t+1)\log(2(t+1))}+2(t+1)}+(t+1)t2^{c2(t+1)\log(2(t+1))}+2(t+1).
\end{equation}
Notice that $q=2^{2^{\Oh(t\log t)}}$.

Consider an instance $(G,I,O,B,\omega,t,T)$ of \probBORDT.

We use the algorithm from Lemma~\ref{lem:unbreakable} for $G$. This algorithm in time $2^{2^{\Oh(t\log t)}}\cdot n^{\Oh(1)}$ either  finds a separation $(U,W)$ of $G$ of order at most $t+1$ such that $|U\setminus W|>q$ and $|W\setminus U|>q$ or correctly reports that $G$ is $((2q+1)q\cdot 2^{t+1},t+1)$-unbreakable. In the latter case we solve the problem using Lemma~\ref{lem:unbreak} in time $2^{2^{\Oh(t\log t)}}\cdot n^{\Oh(1)}$. 
Assume from now that there is a separation  $(U,W)$ of order at most $t+1$ such that $|U\setminus W|>q$ and $|U\setminus W|>q$.

Recall that $|T|\leq 2(t+1)$. Then $|T\cap (U\setminus W)|\leq t+1$ or $|T\cap (W\setminus U)|\leq t+1$. Assume without loss of generality that $|T\cap (W\setminus U)|\leq t+1$. Let $\tilde{G}=G[W]$, $\tilde{I}=I\cap W$, $\tilde{O}=O\cap W$, $\tilde{\omega}$ is the restriction of $w$ to $W$, and define $\tilde{T}=(T\cap W)\cup (U\cap W)$. Since $|U\cap W|\leq t+1$, $|\tilde{T}|\leq 2(t+1)$.

If $|W|\leq (2q+1)q\cdot 2^{t+1}$, then we solve \probBORDF for the instance $(\tilde{G},\tilde{I},\tilde{O},\tilde{B},\tilde{\omega},t,\tilde{T})$ by brute force in time $2^{2^{\Oh(t\log t)}}$ trying all possible subset of $W$ at most $t+1$ values of $0\leq t'\leq t$. Otherwise, we solve  $(\tilde{G},\tilde{I},\tilde{O},\tilde{B},\tilde{\omega},t,\tilde{T})$ recursively. Let $\mathcal{R}$ be the set of nonempty induced subgraphs $R$ that are included in the obtained solution for $(\tilde{G},\tilde{I},\tilde{O},\tilde{B},\tilde{\omega},t,\tilde{T})$. 

For $R\in \mathcal{R}$, define $S_R$ to be the set of vertices of $W\setminus V(R)$ that are adjacent to the vertices of $R$ in the graph obtained by the border complementation for which $R$ is a solution of the corresponding instance of \probMAXT. Note that $|S_R|\leq t$.
If $\mathcal{R}\neq\emptyset$, then let $S=\tilde{T}\cup_{R\in\mathcal{R}}S_R$, and $S=\tilde{T}$ if $\mathcal{R}=\emptyset$.
Since \probMAXT is solved for at most $t+1$ of values of $t'\leq t$, at most $2^{c(2(t+1))\log (2(t+1)) }$ feasible 4-tuples  $(X,Y,Z,\mathcal{P})$, we have that $|\mathcal{R}|\leq (t+1)2^{c2(t+1)\log (2(t+1))}$. Taking into account that $|T'|\leq 2(t+1)$, 
\begin{equation}\label{eq:s-t}
|S|\leq (t+1)t2^{c2(t+1)\log (2(t+1))}+2(t+1).
\end{equation}

Let $\hat{B}=(B\cap U)\cup (B\cap C)$. We claim that the instances $(G,I,O,B,\omega,t,T)$ and \linebreak $(G,I,O,\hat{B},\omega,t,T)$ of \probBORDT are equivalent. 

Recall that we have to show that  
\begin{itemize}
\item[(i)] $T\cap B=T\cap \hat{B}$,
\item[(ii)] for the border complementations $(G',I',O',B',\omega',t')$ and $(G',I',O',\hat{B}',\omega',t')$ of the instances $(G,I,O,B,\omega,t')$ and $(G,I,O,\hat{B},\omega,t')$ respectively of \probMAXF with respect to every feasible $(X,Y,Z,\mathcal{P})$ and $t'\leq t$, it holds that if $(G',I',O',B',\omega',t')$ has a nonempty solution $R_1$, then \linebreak $(G',I',O',\hat{B}',\omega',t')$ has a nonempty solution $R_2$ with $\omega'(V(R_2))\geq \min\{\omega'(V(R_1)),w\}$ and, vice versa, if $(G',I',O',\hat{B}',w',t')$ has a nonempty solution $R_2$, then $(G',I',O',B',w',t')$ has a nonempty solution $R_1$ with $\omega'(V(R_1))\geq \min\{\omega'(V(R_2)),w\}$.
\end{itemize}

The condition (i) holds by the definition of $\hat{B}$. Because $\hat{B}\subseteq B$,  we immediately obtain if $(G',I',O',\hat{B}',\omega',t')$ has a nonempty solution $R_2$, then $(G',I',O',B',\omega',t')$ has a nonempty solution $R_1$ with $\omega'(V(R_1))\geq \min\{\omega'(V(R_2)),w\}$. 
It remains to prove that  for a border complementation $(G',I',O',B',\omega',t')$ and $(G',I',O',\hat{B}',\omega',t')$ of $(G,I,O,B,\omega,t')$ and $(G,I,O,\hat{B},\omega,t')$ respectively of \probMAXF with respect to a feasible $(X,Y,Z,\mathcal{P})$ and $t'\leq t$, it holds that if $(G',I',O',B',\omega',t')$ has a nonempty solution $R_1$, then $(G',I',O',\hat{B}',\omega',t')$ has a nonempty solution $R_2$ with $\omega'(V(R_2))\geq \min\{\omega'(V(R_1)),w\}$.

If $V(R_1)\cap V(G)\subseteq U\setminus W$, then $N_{G'}V(R_1)\subseteq \hat{B}'$. Therefore, for a solution $R_2$ of \linebreak $(G',I',O',\hat{B}',\omega',t')$, $\omega'(V(R_2))\geq \min\{\omega'(V(R_1),w)\}$. 
Assume that $V(R_1)\cap W\neq \emptyset$. 
Let $\tilde{X}=\tilde{T}\cap(V(R_1)\cap W)$, let $\tilde{Y}$ be the set of vertices of $\tilde{T}\setminus V(R_1)$ that are adjacent to vertices of $R_1$ laying outside  $W\setminus U$ and $\tilde{Z}=\tilde{T}\setminus (\tilde{X}\cup \tilde{Y})$. 
Notice that $R_1-(W\setminus U)$ is  a forest and denote it by $F$. Then there is a partition $\tilde{\mathcal{P}}=(P_1,\ldots,P_s)$ of $\tilde{X}$ into nonempty  sets such that two vertices of $\tilde{X}$ are in the same set $P_i$ if and only if they are in the same component of $F$.
Consider the border complementation of $(\tilde{G},\tilde{I},\tilde{O},\tilde{B},\tilde{\omega},\tilde{t})$ with respect to $(\tilde{X},\tilde{Y},\tilde{Z},\tilde{\mathcal{P}})$, where $\tilde{t}$ is the number of neighbors of $R_1$ in $W$. Recall that in the border complementation we have new vertices $u_1,\ldots,u_s\in \tilde{I}$ such that each $u_i$ is adjacent to one vertex of $P_i$ and the vertices of $\tilde{Y}$ for $i\in\{1,\ldots,s\}$. Consider the subgraph $F'$ of $\tilde{G}$ induced by $(V(R_1)\cap W)\cup\{u_1,\ldots,u_s\}$. It is straightforward to see  that $F'$ is a tree that has $\tilde{t}$ neighbors in $\tilde{G}$. 
It implies that there is $\tilde{R}\in \mathcal{R}$ for the instance 
$(\tilde{G},\tilde{I},\tilde{O},\tilde{B},\tilde{\omega},\tilde{t})$ of \probMAXT obtained by the border complementation with respect to $(\tilde{X},\tilde{Y},\tilde{Z},\tilde{\mathcal{P}'})$ and $\tilde{\omega}(V(\tilde{R}))\geq\min\{\tilde{\omega}(V(R_1)\cap W)\}$. 
Recall also that the neighbors of the vertices of $\tilde{R}$ are in $S$.
Now let $R_2=G'[(V(R_1)\cap U)\cup (V(\tilde{R})\cap W)]$. We have that $R_2$ is a tree of weight  at least $\min\{\omega'(R_1),w\}$ that has at most $t'$ neighbors in $G'$.

Since, $(G,I,O,B,\omega,t,T)$ and $(G,I,O,\hat{B},\omega,t,T)$ of \probBORDF are equivalent, we can consider $(G,I,O,\hat{B},\omega,t,T)$. Now we apply some reduction rules that produce equivalent instances of \probBORDF or report that we have no solution. The aim of these rules is to reduce the size of $G$.

Let $Q$ be a component of $G[W]-S$. Notice that for any nonempty graph  $R$ in a solution of $(G,I,O,\hat{B},\omega,t,T)$, either $V(Q)\subseteq V(R)$ or $V(Q)\cap V(R)=\emptyset$, because $N_{G[W]}(V(R))\subseteq S$. Moreover, if
$V(Q)\cap V(R)=\emptyset$, then $N_{G[W]}[V(Q)]\cap V(R)=\emptyset$. 
Notice also that if $v\in N_{G[W]}(V(Q))$ is a vertex of $R$, then  $V(Q)\subseteq V(R)$.
These observation are crucial for the following reduction rules.

\begin{reduction}\label{red:one-rec-t}
For a component $Q$ of $G[W]-S$ do the following in the given order:
\begin{itemize}
\item if $N_{G[W]}[V(Q)]\cap I\neq\emptyset$ and $V(Q)\cap O\neq\emptyset$, then return $\emptyset$ and stop,
\item if $N_{G[W]}[V(Q)]\cap I\neq\emptyset$, then set $I=I\cup V(Q)$,
\item if $V(Q)\cap O\neq\emptyset$, then set $O=O\cup N_{G[W]}[V(Q)]$.
\end{itemize}
\end{reduction}

The rule is applied to each component $Q$ exactly once. Notice that after application of the rule, for every component $Q$ of $G[W]-S$, we have that either $V(Q)\subseteq I$ or $V(Q)\subseteq O$ or $V(Q)\cap (I\cup O\cup\hat{B})=\emptyset$.

Observe that if a component $Q$ of $G[W]$ contains a cycle, then $Q$ cannot be a part of any solution. It leads us to the following rule that is applied to each component.

\begin{reduction}\label{Ared:two-rec-t}
If for a component $Q$ of $G[W]-S$, $Q$ contains a cycle, then 
\begin{itemize}
\item if $V(Q)\subseteq I$, then return $\emptyset$ and stop, otherwise,
\item set $O=O\cup N_{G[W]}[V(Q)]$.
\end{itemize}
\end{reduction}

Suppose that there is a component $Q$ containing two distinct verices $u$ and $v$ that are adjacent to the same vertex $x\in N_{G[W]}(V(Q))$. If there is a graph $R$ in a solution of $(G,I,O,\hat{B},\omega,t,T)$ with $u,v\in V(R)$, then $x\notin V(R)$, because $R$ is a tree. If $u,v\notin V(R)$, then $x\notin V(R)$, because otherwise $u,v\notin \hat{B}$ would be adjacent to a vertex of $R$. Hence, the next rule is safe. 

\begin{reduction}\label{red:three-rec-t}
If for a component $Q$ of $G[W]-S$, there is $x\in N_{G[W]}(Q)\setminus O$ adjacent to two distinct vertices of $Q$, then
set $O=O\cup\{x\}$.
\end{reduction}

We apply the rule exhaustively while it is possible.

After applying Reduction Rules~\ref{red:one-rec-t}-\ref{red:three-rec-t}, we can safely replace each $Q$ by a single vertex.

\begin{reduction}\label{red:four-rec-t}
If for a component $Q$ of $G[W]-S$, $|V(Q)|>1$, then 
\begin{itemize}
\item modify $G$ by deleting the vertices of $V(Q)$ and constructing a new vertex $u$ adjacent to $N_{G[W]}(V(Q))$,
\item set $W=(W\setminus V(Q))\cup \{u\}$,
\item set $\omega(u)=\omega(V(Q))$,
\item if $V(Q)\subseteq I$, then set $I=(I\setminus V(Q))\cup\{u\}$,
\item if $V(Q)\subseteq O$, then set $O=(O\setminus V(Q))\cup\{u\}$.
\end{itemize}
\end{reduction}

The rule is applied for each $Q$ with at least two vertices. Notice that now $W\setminus S$ is an independent set. 

Note that if there is  $u\in W\setminus S$ such that $u\in O$, then it is safe to delete $u$, because $N_{G[W]}(u)\subseteq O$ by Reduction Rule~\ref{red:one-rec-t}.

\begin{reduction}\label{red:five-rec-t}
If there is $u\in W\setminus S$ such that $u\in S$, then delete $u$ from $G$ and the sets $W$ and $O$.
\end{reduction}

From now we have that $(W\setminus S)\cap O=\emptyset$.

Suppose that $L$ is a set of verices of  $W\setminus S$ that have the same neighborhoods, i.e, they are false twins of $G[W]$. If there are distinct $u,v\in L$ such that $u,v\in V(R)$ for 
a graph $R$ in a solution of $(G,I,O,\hat{B},\omega,t,T)$, then $L\subseteq V(R)$, because $R$ should contain exactly one vertex in the neighborhoods of $u$ and $v$.
Hence, either $L\cap V(R)=\emptyset$ or exactly one vertex of $L$ is in $R$ or $L\subseteq V(R)$. In particular, if $|L\cap I|\geq 2$, then $L\subseteq V(R)$. 
Suppose $L\cap I=\emptyset$ and 
$u\in L$ is the unique vertex of $L$ in $R$, then we can safely assume that $u$ is a vertex of maximum weight in $L$.  
Notice also that in this case $u$ is the unique vertex of $R$. It means that $R$ is obtained for a border complementation of  $(G,I,O,\hat{B},\omega,t)$ with respect to $(\emptyset,\emptyset,T,\emptyset)$ and we output $R$ only if $\omega(u)\geq w$.
These observations give us the two following rules that are applied for all inclusion maximal sets 
 $L\subseteq W\setminus S$ of size at least 3.

\begin{reduction}\label{red:six-rec-t}
If for an inclusion maximal set of false twin vertices $L\subseteq W\setminus S$ such that $|L|\geq3$, $L\cap I\neq\emptyset$, then let $u\in L\cap I$, $v\in L\setminus \{u\}$, and 
\begin{itemize}
\item delete the verices of $L\setminus\{u,v\}$ from $G$ and the sets $W,I$,
\item if $(L\setminus \{u\})\cap I\neq \emptyset$, then set $I=I\cup \{v\}$,
\item set $\omega(v)=\sum_{x\in L\setminus\{u\}}\omega(x)$.
\end{itemize}
\end{reduction}

\begin{reduction}\label{red:seven-rec-t}
If for an inclusion maximal set of false twin vertices $L\subseteq W\setminus S$ such that $|L|\geq3$, $L\cap I=\emptyset$, then let $u$ be a vertex of maximum weight in $L$ and let $v\in L\setminus \{u\}$, and then 
\begin{itemize}
\item delete the verices of $L\setminus\{u,v\}$ from $G$ and the set $W$,
\item set $\omega(v)=\min\{\sum_{x\in L\setminus\{u\}}\omega(x),w-1\}$.
\end{itemize}
\end{reduction}

Notice that $\omega(v)\leq w-1$. It implies that $v$ cannot be selected as a unique vertex of a solution.
To see that it is safe, observe that if $\omega(u)>0$, then the total weight of $u$ and $v$ is at least $w$ and recall that by the definition of \probMAXT, we output nonempty graphs of maximum weight or weight at least $w$. Therefore, if we output $R'$ that includes $u$ and $v$ instead of $R$ containing all the vertices of $L$, then we output a graph of weight at least $\min\{\omega(R),w\}$. 

Denote by $(G^*,I^*,O^*,B^*,\omega^*,t,T)$ the instance of \probBORDF obtained from  $(G,I,O,\hat{B},\omega,t,T)$ by Reduction Rules~\ref{red:one-rec-t}-\ref{red:seven-rec-t}. Notice that all modifications were made for $G[W]$. Denote by $W^*$ the set of vertices of the graph obtained from $G[W]$ by the rules. Observe that there are at most $2^{|S|}$ subsets $L$ of $S$ such that there is a components $Q$ of $G[W]-S$ with $N_{G[W]}(V(Q))=L$. Notice that for every $L$ all such components $Q$ are replaced by at most 2 vertices by the reduction rules.
Taking into account the vertices of $S$, we obtain the following upper bound for the size of $W^*$:
\begin{equation}\label{eq:w*-t}
|W^*|\leq 2\cdot 2^{|S|}+|S|.
\end{equation}
By (\ref{eq:q-t}) and (\ref{eq:s-t}), $|W^*|\leq q$.
Recall that $|W\setminus U|>q$. Therefore, $|V(G*)|<|V(G)|$. We use it and solve \probBORDT for $(G^*,I^*,O^*,B^*,\omega^*,t,T)$ recursively.

To evaluate the running time, denote by $\tau(G,I,O,B,\omega,t,T)$ the time needed to solve \probBORDT for $(G,I,O,B,\omega,t,T)$.
Clearly, all reduction rules are polynomial. The algorithm from Lemma~\ref{lem:unbreakable} runs in time $2^{2^{\Oh(t\log t)}}\cdot n^{\Oh(1)}$. Then we obtain the following recurrence for the running time:
\begin{equation}\label{eq:run-t}
\tau(G,I,O,B,\omega,t,T)\leq \tau(G^*,I^*,O^*,B^*,\omega^*,t,T)+\tau(\tilde{G},\tilde{I},\tilde{O},\tilde{B},\tilde{\omega},t,\tilde{T})+2^{2^{\Oh(t\log t)}}\cdot n^{\Oh(1)}.
\end{equation}
Because $|W^*|\leq q$,
\begin{equation}\label{eq:vert-t}
|V(G^*)|\leq |V(G)|-|V(\tilde{G})|+q.
\end{equation}
Recall that if the algorithm of Lemma~\ref{lem:unbreakable} reports that $G$ is $((2q+1)q\cdot 2^{t+1},t+1)$-unbreakable or we have that $|V(G)|\leq ((2q+1)q\cdot 2^{t+1}$, we do not recurse but solve the problem directly in time $2^{2^{\Oh(t\log t)}}\cdot n^{\Oh(1)}$. Following the general scheme from~\cite{ChitnisCHPP16}, we obtain that these condition together with (\ref{eq:run-t}) and (\ref{eq:vert-t}) imply that the total running time is  $2^{2^{\Oh(t\log t)}}\cdot n^{\Oh(1)}$.
\end{proof}

We have now all the details that are necessary to prove the main theorem of this section.

\begin{theorem}\label{thm:tree}
\probLT can be solved in time  $2^{2^{\Oh(t\log t)}}\cdot n^{\Oh(1)}$.
\end{theorem}

\begin{proof}
Let $(G,\omega,t,w)$ be an instance of \probLT. We define $I=\emptyset$, $O=\emptyset$, $B=V(G)$ and $T=\emptyset$. Then we solve \probBORDT for $(G,I,O,B,\omega,t,T)$ using Lemma~{lem:bordtree}
 in time $2^{2^{\Oh(t\log t)}}\cdot n^{\Oh(1)}$. It remains to notice that $(G,\omega,t,w)$ is a yes-instance of \probLT if and only if $(G,I,O,B,\omega,t,T)$ has a nonempty graph in a solution.
\end{proof}

\section{Better algorithms for \probLCSS}
\label{sec:better}
We applied the recursive understanding technique introduced by Chitnis et al.~\cite{ChitnisCHPP16} for \probLCSS when $\Pi$ is defined by a finite set of forbidden subgraphs and $\Pi$  is the property to be a forest in Sections~\ref{sec:fpt-forb} and \ref{sec:fpt-tree} respectively. We believe that the same approach can be used in some other cases. In particular,  
the results of Section~\ref{sec:fpt-tree} could be generalized if $\Pi$ is the property defined by a finite list of \emph{forbidden minors} that include a planar graph.
Recall that a graph $F$ is a \emph{minor} of a graph $G$ if a graph isomorphic to $F$ can be obtained from $G$ by a sequence of vertex and edge deletion and edge contractions.
Respectively, a graph $G$ is $F$-minor free if $F$ is not a minor of $G$, and for a family of graphs $\mathcal{F}$,  $G$ is $\mathcal{F}$-minor free if $G$ is $F$-minor free for every $F\in\mathcal{F}$. If $F$ is a planar graph, then by the fundamental results of Robertson and Seymour~\cite{RobertsonS86} (see also~\cite{RobertsonST94}), an $F$-minor free graph $G$ has bounded \emph{treewidth}. This makes it possible to show that if $\Pi$ is the property to be $\mathcal{F}$-minor free for a finite family of graphs $\mathcal{F}$ that includes at lease one planar graph, then \probLCSS is \classFPT when parameterized by $t+k$.
Nevertheless, the drawback of the applying the recursive understanding technique for these problems is that we get double or even triple-exponential dependence on the parameter in our \classFPT algorithms. It is natural to ask whether we can do better for some properties $\Pi$. In this section we show that it can be done if $\Pi$ is the property to be a complete graph,  a star, to be $d$-regular or to  be a path.

\subsection{Secluded Clique}

We begin with problem \probSC, defined as follows.

 \defproblema{\probSC}%
{A graph $G$, a weight function $\omega\colon V(G)\rightarrow \mathbb{Z}_{>0}$, an nonnegative integer $t$ and a positive integer $w$.}%
{Decide whether $G$ contains a $t$-secluded clique $H$ with $\omega(V(H))\geq w$.}

We prove that this problem can be solved in time $2^{\cO(t \log t)} \cdot n^{\cO(1)}$. The result uses the algorithm of Lemma \ref{lem:derand} and the following simple observations.  

\begin{lemma}\label{lem:SCtruetwins}
Let $(G,\omega, t, w)$ be an input of \probSC and let $H$ be a solution such that $V(H)$ is maximal by inclusion.  Let $L$ be an inclusion maximal set of true-twins of $G$. Then $L\cap V(H) \neq \emptyset$ implies that $L \subseteq V(H)$.
\end{lemma}

\begin{proof}
Let $L$ be an inclusion maximal set of twins, and let $u,v \in L$ be such that $u\in V(H)$ and $v \notin V(H)$. Consider the graph $H' = G[V(H) \cup \{v\}]$. Since $H$ is a $t$-secluded clique, and $u, v$ are true twins, we have that $H'$ is also a $t$-secluded clique, and $\omega(V(H')) = \omega(V(H)) + \omega(v) \geq w$. Therefore  $H'$ is also a solution of \probSC, contradicting the maximality of $H$. 
\end{proof}

Let $\mathcal{L}$ be the family all of maximal sets of true twins in a graph $G$. Note that a vertex can not belong to two different maximal sets of true twins, so $\mathcal{L}$ induces a partition of $G$. Consider $\tilde{G}$ the graph obtained from $G$ contracting each maximal set of true-twins $L$ into a single node $x_L$, and removing multiple edges. In other words, $\tilde{G}$ contains one node for each element of $\mathcal{L}$. Two nodes $x_1$ and $x_2$ in $\tilde{G}$ are adjacent if there is an edge in $G$ with one endpoint in $L_1$ and the other one in $L_2$, where $L_1$ and $L_2$ are elements of $\mathcal{L}$ corresponding to $x_1$ and $x_2$, respectively.  

 We say that a node $x \in \tilde{G}$ is a \emph{contraction} of $L \in \mathcal{L}$ if $x$ is the node of $\tilde{G}$ corresponding to $L$.  We say that a set $U \subseteq V(G)$ is the \emph{expansion} of~$\tilde{U} \subseteq V(\tilde{G})$ if the all the nodes of $\tilde{U}$ represent maximal sets of true-twins contained in $U$. We also say in that case that $\tilde{U}$ is the \emph{contraction} of $U$. 

\begin{lemma}\label{SCsize}
Let $G$ be a graph, $t$ be a positive integer and let $U$ be a subset of vertices of $G$ inducing a inclusion maximal $t$-secluded clique in $G$. There exists a set $\tilde{U}$ of vertices of $V(\tilde{G})$ such that:
\begin{itemize}
\item[1)] $U$ is the expansion of $\tilde{U}$, 
\item[2)] $\tilde{U}$ induces a $t$-secluded clique on $\tilde{G}$, and
\item[3)] $|\tilde{U}| \leq 2^t$.
\end{itemize} 

\end{lemma}

\begin{proof}
 Let $\tilde{U}$ be the set of nodes consisting in the contraction of the maximal sets of true-twins in $G$ intersecting $U$. We claim that $\tilde{U}$ satisfies the desired properties. 
\begin{itemize}
\item[1)] From Lemma \ref{lem:SCtruetwins}, we know that if a maximal set of true-twins $L$ intersects $U$, then it $L$ is contained in $U$. Therefore $\tilde{U}$ is a contraction of $U$ (so $U$ is an expansion of $\tilde{U}$). 

\item[2)] Let $x_1$ and $x_2$ be two nodes in $\tilde{U}$ that are contractions of $L_1$ and $L_2$, respectively. Since $U$ induces a clique in $G$, $L_1$ and $L_2$ must contain adjacent vertices, so $x_{1}$ and $x_{2}$ are adjacent in $\tilde{G}[\tilde{U}]$.  On the other hand, $|N_{\tilde{G}}(\tilde{U})|$ equals the number of maximal sets of true-twins intersecting $N_G(U)$, so $|N_{\tilde{G}}(\tilde{U})|\leq |N_G(U)| \leq t$. We conclude that $\tilde{U}$ induces a $t$-secluded clique in $\tilde{G}$. 

\item[3)] Let $x_1$ and $x_2$ be two different nodes in $\tilde{U}$ that are contractions of $L_1$ and $L_2$, respectively. From definition of maximal sets of true-twins, $N_G(L_1) \neq N_G(L_2)$, so $N_{\tilde{G}}(x_1) \neq N_{\tilde{G}}(x_2)$. Since $\tilde{U}$ is a clique, necessarily $N_{\tilde{G}}(x_1) \cap \tilde{U} = N_{\tilde{G}}(x_2) \cap \tilde{U}$.  Therefore, every vertex on $\tilde{U}$ has a different neighborhood outside $\tilde{U}$. Since $|N_{\tilde{G}}(\tilde{U})| \leq t$, we obtain that $|\tilde{U}| \leq 2^{|N_{\tilde{G}}(\tilde{U})|} \leq 2^t$. 
\end{itemize}
\end{proof}

\begin{theorem}
\probSC can be solved in time $2^{\cO(t \log t)} \cdot n^{\cO(1)}$.
\end{theorem}

\begin{proof}

The algorithm for \probSC on input $(G,\omega, t, w)$ first computes the family $\mathcal{L}$ of all inclusion maximal set of true-twins of $G$, and then computes $\tilde{G}$ in time $\mathcal{O}(n^2)$. Then, the algorithm uses Lemma \ref{lem:derand} to compute in time $2^{\cO(t \log t)}n \log n$ a family $\mathcal{S}$ of at most $2^{\cO(t \log t)}\log n$ subsets of $V(\tilde{G})$ such that: for any sets $A, B \subseteq V(\tilde{G})$, $A \cap B = \emptyset$, $|A| \leq 2^t$, $|B| \leq t$, there exists a set $S \in \mathcal{S}$ with $A\subseteq S$ and $B\cap S = \emptyset$. 

Let $U$ be a set of vertices of $G$ inducing an inclusion maximal solution of \probSC on instance  $(G,\omega, t, w)$, and let $\tilde{U}$ be the contraction of $U$. From Lemma \ref{SCsize}, we know that $N_{\tilde{G}}(\tilde{U}) \leq t$ and $|\tilde{U}| \leq 2^t$.  Then, there exists $S\in \mathcal{S}$ such that $\tilde{U} \subseteq S$ and $N_{\tilde{G}}(U) \cap S = \emptyset$.  In other words $\tilde{U}$ is a component of $\tilde{G}[S]$. Therefore, the algorithm checks for every $S \in \mathcal{S}$ and every component $C$ of $\tilde{G}[S]$ if the expansion of $C$ is a solution of the problem. 
\end{proof}

\subsection{Secluded Star}

Another example of a particular problem where we have better running times is \probSSt, defined as follows.

\defproblema{\probSSt}%
{A graph $G$, a weight function $\omega\colon V(G)\rightarrow \mathbb{Z}_{>0}$,  an nonnegative integer $t$,  and a positive integer $w$.}%
{Decide whether $G$ contains a  $t$-secluded induced star $S$ with $\omega(V(S))\geq w$.}

In this case, a faster FPT algorithm can be deduced via a reduction to the problem \textsc{Vertex Cover} parameterized by the size of the solution.

For a graph $G$ and $x$ in $V(G)$, we call $N^2(x)$ the set of vertices at distance $2$ from $x$, i.e., $N^2(x)$ is the set of vertices $u \in V(G) $ such that $ u \notin N[x]$ and there exists $ v \in N(x)$ such that  $ u \in N(v)$. We also call $N^2[x]$ the set $N^2(x) \cup N[x]$.  Let now $F_x = (N^2(x)  \cup N(x), E')$ be the subgraph of $G[N^2(x) \cup N(x)]$ such that $E' = E(G[N^2(x)\cup N(x)]) - E(G[N^2(x)])$, i.e., $F_x$ is the graph induced by the vertices in $N^2(x) \cup N(x)$ after the deletion of all edges between nodes in $N^2(x)$. Note that $x$ is not a vertex of $F_x$.

 A vertex $x$ is the center of a star $S$ if $x$ is the vertex of maximum degree in $S$.
 The following lemma relates the center $x$ of a $t$-secluded star $S$ of a graph $G$, with a vertex cover of size at most $t$ of $F_x$.

\begin{lemma}\label{lem:SStVC}
Let $S$ be a $t$-secluded star on a graph $G$ and let $x$ be the center of $S$. Then $N_G(V(S))$ is a vertex cover of $F_x$. Moreover, if $S$ is an inclusion maximal $t$-secluded star with the center $x$, then $N_G(V(S))$ is an inclusion minimal vertex cover of $F_x$.
\end{lemma}

\begin{proof}
Let $u,v$ two adjacent vertices of $F_x$. Since $F_x[N^2(x)]$ is edgeless, we assume w.l.o.g. that $u$ belongs to $N_G(x)$. If $u$ is contained in $N_G(x) \setminus V(S)$, then $u$ is in $N(V(S))$ (because $x$ belongs to $S$, so $N_G(x) \setminus V(S) \subseteq N_G(V(S))$). If $u$ is in $S$, then either $v$ is in $N_G^2(x)$ or $v$ is in $N_G(x) \setminus V(S)$ (because $F_x[V(S)]$ is edgeless). In both cases $v$ is in $N_G(V(S))$. We conclude that either $u$ or $v$ is contained in $N_G(V(S))$. 

Assume that $S$ is an inclusion maximal $t$-secluded star with the center $x$. Suppose that, contrary to the second claim,  $N_G(V(S))$ is not an inclusion minimal vertex cover of $F_x$, that  is, there is $u\in N_G(V(S))$ such that $X=N_G(V(S))\setminus\{u\}$ is a vertex cover of $F_x$. Because $X$ is a vertex cover of $F_x$ and $X\subseteq N_G(V(S))$, we have that $u\in N_G(x)$ and $N_G(u)\setminus \{x\}\subseteq X$. It implies that $S\rq{}=G[V(S)\cup \{u\}]$ is a $t$-secluded star contradicting the maximality of $S$.      
\end{proof}

A basic result on parameterized complexity states that it could decided if a graph contains a vertex cover of size at most $t$ in time 
$\Oh(c^t(n+m))$ for some constant $c$ (see, e.g., \cite{CyganFKLMPPS15,DowneyF13})  by branching algorithms.  These algorithms could be adapted to output the list of all inclusion minimal vertex covers of size at most $t$ within the same running time. In fact, it could be done for $c<2$ but this demands some discussion. Hence, we use the following straightforward observation.

\begin{proposition}[\cite{CyganFKLMPPS15,DowneyF13}]\label{prop:vc-enum}
There is an algorithm computing the list of all the inclusion minimal vertex covers of size at most $t$ of a graph $G$ in time $\Oh(2^t(n+m))$.
\end{proposition}

\begin{theorem}
\probSSt can be solved in time $\Oh(2^t \cdot n^{\cO(1)})$.
\end{theorem}

\begin{proof}
Let $(G,\omega, t, w)$ be an input of \probSSt. The algorithm starts computing, for every $x \in V(G)$, the list of all inclusion minimal vertex covers of size at most $t$ of $F_x$ using Proposition~\ref{prop:vc-enum}. Then, for every vertex cover $U$  of size at most $t$ of $F_x$, the algorithm checks if $N_G[x]\setminus U$ induces in $G$ a solution of the problem. We know from Lemma \ref{lem:SStVC}, that if $S$ is a solution of \probSSt on input $(G, \omega, t, w)$, then  $N_G(V(S))$ is an inclusion minimal  vertex cover of size at most $t$ of $F_x$. Note also that $S$ is an induced star with center $x$ in a graph $G$, then $V(S) = N_G[x] \setminus N_G(V(S))$. 
\end{proof}

\subsection{Secluded Regular Subgraph}

Another example of a problem with single-exponential FPT algorithm is \probCSDRS, defined as follows.

 \defproblema{\probCSDRS}%
{A graph $G$, a weight function $\omega\colon V(G)\rightarrow \mathbb{Z}_{>0}$,  an nonnegative integer $t$,  and positive integers $w$ and $d$.}%
{Decide whether $G$ contains a connected $t$-secluded $d$-regular induced subgraph $H$ with $\omega(V(H))\geq w$.}

Let $(G,\omega,t,w,d)$ be an input of \probCSDRS and let $U$ be a set of vertices of $G$ such that $G[U]$ is a solution of the problem. Note first that any vertex of degree greater than $t+d$ can not be contained in $U$, otherwise $G[U]$ is not $t$-secluded. Let $W$ be the set of vertices of high degree, i.e., $x\in W$ if $|N(x)| \geq t+d+1$. Suppose that $N(U) \subseteq W$. This implies that $U$ is a component of $G-W$. Therefore, our algorithm will first compute the components of $G-W$ and check if some of them is a solution. In the following we assume that $N(U) \setminus W \neq \emptyset$. 

We call $L = N(U)\setminus W$, $U_1 = N(L) \cap U$, $U_2 = N(U_1) \cap U$ and $\tilde{U} = U_1 \cup U_2$. Note that $|L| \leq t$, $|U_1| \leq t\cdot (t+d)$ and $|U_2| \leq  d t \cdot (t+d)$. Therefore $|\tilde{U}| \leq t(d+1)(t+d)$.

A set of vertices $C$ is called \emph{good} for $U$ if 
\begin{itemize}
\item $C \subseteq U$, and
\item For all $u \in C$, $u\in U_1$ implies $N(u)\cap U \subseteq C$.  
\end{itemize}

Note that every node $u$ in a good set $C$ satisfies $|N(u)\cap C| \leq d$. Moreover if $u \in U_1 \cap C$ then $|N(u)\cap C| = d$. Note also that if $C_1$ and $C_2$ are good for $U$ then $C_1 \cup C_2$ is good for $U$.

\begin{lemma}\label{lem:dregulargrow}

Let $S$ be a set of vertices of $G$ satisfying $\tilde{U} \subseteq S$ and $S\cap L = \emptyset$. Let now $C$ be a component of $G[S]$ such that $C \cap U \neq \emptyset$, then:

\begin{itemize}
\item[1)] $C$ is good for $U$, and
\item[2)] if $u \in C$ is such that $|N(u) \cap C| < d$, then $S' = S \cup N(u)$ satisfies $\tilde{U} \subseteq S'$ and $S'\cap L = \emptyset$.  
\end{itemize}

\end{lemma}

\begin{proof} \hspace{1cm}

\begin{itemize}
\item[1)] Let $u \in C \cap U$ and $v \in N(u) \cap (C \setminus U)$. Then $v$ is contained in $L$, which contradicts the fact that $S\cap L = \emptyset$. Therefore $C \subseteq U$. On the other hand, if $u \in U_1 \cap C$ then $N(u) \cap U \subseteq C$, because $U_2$ is contained in $S$ and $C$ is a connected component of $G[S]$.  We conclude that $C$ is good for $U$.

\item[2)] Let $u \in C$ be such that $|N(u) \cap C| < d$. Since $C$ is good for $U$ we know that $u$ is not contained in $U_1$, so $N(u) \cap L = \emptyset$.  Therefore $\tilde{U} \subseteq S \subseteq S'$ and $S' \cap L = \emptyset$.  
\end{itemize}
\end{proof}

\begin{theorem}
\probCSDRS can be solved in time $2^{\cO(t\log (td))} \cdot n^{\cO(1)}$.
\end{theorem}

\begin{proof}

The algorithm for \probCSDRS on input $(G,\omega, t, w, d)$ first computes the set $W$ of vertices $v$ satisfying $|N(v)| > t+d+1$. Then computes the connected components of $G-W$ and checks if some of them is a solution. If a solution is not found this way, the algorithm computes $\tilde{G} = G-W$. Then, the algorithm uses Lemma \ref{lem:derand} to compute in time
$2^{\cO(t \log (td))}n \log n$ 
a family $\mathcal{S}$ of at most $2^{\cO(t\log (td))}\log n$
 subsets of $V(\tilde{G})$ such that: for any sets $A, B \subseteq V(\tilde{G})$, $A \cap B = \emptyset$, $|A| \leq t(d+1)(t+d)$, $|B| \leq t$, there exists a set $S \in \mathcal{S}$ with $A\subseteq S$ and $B\cap S = \emptyset$. 

For each set $S \in \mathcal{S}$, the algorithm mark as \emph{candidate} every component $C$ of $\tilde{G}[S]$ that satisfy that for all $u\in C$, $|N(u) \cap C| \leq d$. For each candidate component $C$, the algorithm looks for a node $u$ in $C$ such that $|N(u) \cap C| <d$. If such node is found, the corresponding component is enlarged adding to $C$ all nodes in $N(u)$. If $N(u)$ intersects other components of $G[S]$, we merge them into $C$. We repeat the process on the enlarged component $C$ until it can not grow any more, or some node $u$ in $C$ satisfies $|N(u) \cap C| >d$.  In the first case we check if the obtained component is a solution of the problem, otherwise the component is unmarked (is not longer a candidate), and the algorithm continues with another candidate component of $\tilde{G}[S]$ or other set $S'\in \mathcal{S}$.

Let $U$ be a set of vertices such that $G[U]$ is a solution of  \probCSDRS on input $(G,\omega, t, w, d)$ such that $N(U) \setminus W \neq \emptyset$. From construction of $\mathcal{S}$, we know that there exists some $S \in \mathcal{S}$ such that $\tilde{U} \subseteq S$ and $S \cap L = \emptyset$. From Lemma \ref{lem:dregulargrow} (1), we know that any component of $\tilde{G}[S]$ intersecting $U$ is good for $U$, so it will be marked as candidate. Finally, from Lemma \ref{lem:dregulargrow} (2) we know that when a component that is good for $U$ grows, the obtained component is also good for $U$.
Indeed, if $C$ is a good component of $\tilde{G}[S]$ containing a node $u$ such that $|N(u)\cap C| <d$, then the component of $\tilde{G}[S \cup N(u)]$ containing $C$ is also good for $U$.  We conclude that we correctly find $U$ testing the enlarging process on each component of $\tilde{G}[S]$, for each $S \in \mathcal{S}$. 
\end{proof}

\subsection{Secluded Long Path}

Our last problem is \probSP, defined as follows.

\defproblema{\probSP}%
{A graph $G$, a weight function $\omega\colon V(G)\rightarrow \mathbb{Z}_{>0}$,  an nonnegative integer $t$,  and a positive integer $w$.}%
{Decide whether $G$ contains a $t$-secluded induced path $P$ with $\omega(V(P))\geq w$.}

In this case the reasoning is very similar to the one for problem \probCSDRS when $d=2$. Let $(G,\omega,t,w)$ be an input of \probSP and let $P$ be a set of vertices of $G$ such that $G[P]$ is a solution of the problem. Note first that any vertex of degree greater than $t+2$ can not be contained in $P$, otherwise $G[P]$ is not $t$-secluded. Let $W$ be the set of vertices of high degree, i.e., $x\in W$ if $|N(x)| \geq t+3$. Suppose first that $N(P) \subseteq W$. This implies that $P$ is a component of $G-W$. Therefore, our algorithm will first compute the components of $G-W$ and check if some of them is a solution. In the following we assume that $N(P) \setminus W \neq \emptyset$.

We call again $L = N(P)\setminus W$ and define $P_1 = N(L) \cap P$, $P_2 = N(P_1) \cap P$ and $\tilde{P} = P_1 \cup P_2$. Note that  $|L| \leq t$, $|P_1| \leq t\cdot (t+3)$ and $|P_2| \leq  2t\cdot (t+3)$. Therefore $|\tilde{P}| \leq  3t\cdot (t+3)$.  For this case, it will be enough to define good sets as paths. A path $P'$ is called \emph{good} for $P$ if satisfies the following conditions:
\begin{itemize}
\item $P' \subseteq P$, and
\item For all $u \in P'$, $u\in P_1$ implies $N(u)\cap P \subseteq P'$.  
\end{itemize}

Note that every vertex $u$ in a path $P'$ good  for $P$ satisfies $|N(u)\cap C| \leq 2$. However, unlike the $d$ regular case, where the nodes in $U_1$ had $d$ neighbors in a good set, in this case we may have a node  $u \in P_1 \cap P'$ such that $|N(u)\cap P'| = 1$, when it is an endpoint of~$P'$.

\begin{lemma}\label{lem:pathgrow}

Let $S$ be any set of vertices of $G$ satisfying $\tilde{P} \subseteq S$ and $S\cap L = \emptyset$. Let now $C$ be a component of $G[S]$ such that $P \cap U \neq \emptyset$. Then the following holds

\begin{itemize}
\item[1)] $C$ is good for $U$ (so $C$ is a path).
\item[2)]  Suppose that $\omega(C) < w$. Then there exists an endpoint $u$ of $C$ satisfying $N(u)\setminus C = \{v\}$,  where $v$ is a vertex of $V(G) \setminus S$ contained in $P$ such that $N(v)\setminus W \subseteq P$. 

\end{itemize}

\end{lemma} 

\begin{proof} \hspace{1cm}
\begin{itemize}
\item[1)] Let $u \in C \cap P$ and $v \in N(u) \cap (C \setminus P)$. Then $v$ is contained in $L$, which contradicts the fact that $S\cap L = \emptyset$. Therefore $C \subseteq P$. On the other hand, if $u \in P_1 \cap C$ then $N(u) \cap P = N(u) \cap \tilde{P} \subseteq C \cap S$. Then $N(u)\cap P \subseteq C$.  We conclude that $C$ is good for $P$. 

\item[2)] 
Since  $\omega(C) < w$ and $C$ is good for $P$, $C$ must be strictly contained in $P$. Hence one of the endpoints of $C$ can be extended, i.e.,  $P \cap (N(u_1) \setminus C) \neq \emptyset$ or $P \cap (N(u_2) \setminus C) \neq \emptyset$. Suppose w.l.o.g. that  $v \in  P \cap (N(u_1) \setminus C)$ (so $v \in P \setminus S$) and  that $N(u_1)$ is not contained in $P$. If $u'$ is a vertex in  $N(u_1) \setminus P$, then necessarily $u'$ belongs to $L$. Therefore $u_1$ is in $P_1$, which implies that $v$ is in $\tilde{P}$.  This contradicts the fact that $\tilde{P}$ is contained in $S$. We conclude that $N(u_1) \subseteq P$, so $N(u_1)\setminus C =  \{v\}$. Clearly $v$ is not in $P_1$, again because  otherwise $\tilde{P}$ would not be contained in $S$. We conclude that $N(v)\cap L = \emptyset$, so $N(v) \setminus W \subseteq P$.

\end{itemize}
\end{proof}

\begin{theorem}
\probSP can be solved in time $2^{\cO(t \log t)} \cdot n^{\cO(1)}$.
\end{theorem}

\begin{proof}

The algorithm for \probSP on input $(G,\omega, t, w)$ first computes the set $W$ of vertices  $v$ of $G$ satisfying $|N(v)| > t+2$. Then computes the connected components of $G-W$ and checks if some of them is a solution. If a solution is not found this way, the algorithm computes $\tilde{G} = G-W$.Then, the algorithm uses Lemma \ref{lem:derand} to compute in time $2^{\cO(t \log t)}n \log n$ a family $\mathcal{S}$ of at most $2^{\cO(t \log t)}\log n$ subsets of $V(\tilde{G})$ such that: for any sets $A, B \subseteq V(\tilde{G})$, $A \cap B = \emptyset$, $|A| \leq 3t(t+3)$, $|B| \leq t$, there exists a set $S \in \mathcal{S}$ with $A\subseteq S$ and $B\cap S = \emptyset$. 

For each set $S \in \mathcal{S}$, the algorithm mark as \emph{candidate} every component $C$ of $G[S]$ that is a path and check if one of them is a solution of the problem. Suppose that none of the marked components is a solution. The then algorithm runs a \emph{growing} process on each candidate component $C$ such that $\omega(C) < w$. For each such component $C$, and each endpoint $u$ of $C$, the algorithm checks if $u$ has only one neighbor $v$ outside $C$. If it does, the algorithm add $v$ to the path if $v$ has at most $2$ neighbors not intersecting $N(C)$. Then the growing process is repeated in the new component $C \cup \{v\}$, eventually through the other endpoint of $C$, and through a neighbor of $v$ outside $C$ (the algorithm test both options). If $N(v)$ intersect another path $C'$, we merge $C$ and $C'$.  The process stops when the resulting component $C$ is a path such that $\omega(C) \geq w$ or it can not grow any more. In the first case the algorithm checks if $C$ is a solution of the problem, 
and otherwise continues with another endpoint of the original component $C$, with another candidate component of $G[S]$ or other set $S' \in \mathcal{S}$. 

Let $P$ be a set of vertices such that $G[P]$ is a solution of  \probSP on input $(G,\omega, t, w)$, such that $N(P)\setminus W \neq \emptyset$. From construction of $\mathcal{S}$, we know that there exists some $S \in \mathcal{S}$ such that $\tilde{U} \subseteq S$ and $L \cap S = \emptyset$. From Lemma \ref{lem:pathgrow} (1), we know that any component of $\tilde{G}[S]$ intersecting $P$ is good for $U$, so it will be marked as candidate. We know from Lemma \ref{lem:pathgrow} (2) that each  component $C$ of $\tilde{G}[S]$ that is good for $P$, such that $\omega(C) < w$, can be extended to a node $v\in P$ such that $N(v) \cap L = \emptyset$.  Hence $S' = S \cup \{v\}$ satisfies $\tilde{P} \subseteq S'$ and $S' \cap L = \emptyset$. Then, from Lemma \ref{lem:pathgrow} (1) the component of $\tilde{G}[S']$ containing $C$ is good for $P$. We conclude that we correctly find $P$ testing the growing process on each endpoint of each component of $\tilde{G}[S]$, for each $S \in \mathcal{S}$. 
\end{proof}

\section{Concluding remarks}\label{sec:concl}
Another interesting question concerns kernelization for \probLCSS. We refer to~\cite{CyganFKLMPPS15} for the formal introduction to the kernelization algorithms.  Recall that a \emph{kernelization} for a parameterized problem is a polynomial algorithm that maps each instance $(x,k)$ with the input $x$ and the parameter $k$ to an instance $(x',k')$ such that i) $(x,k)$ is a yes-instance if and only if $(x',k')$ is a yes-instance of the problem, and ii) $|x'|+k'$ is bounded by $f(k)$ for a computable function $f$. The output $(x',k')$ is called a \emph{kernel}. The function $f$ is said to be a \emph{size} of a kernel. Respectively, a kernel is \emph{polynomial} if $f$ is polynomial. 

For \probLCSS, we hardly can hope to obtain polynomial kernels as it could be easily proved by applying the results of Bodlaender et al.~\cite{BodlaenderDFH09}  that, 
unless  $\classNP\subseteq\classCoNP/\text{\rm poly}$,
 \probLCSS has no polynomial kernel when parameterized by $t$  if \probLCSS is \classNP-complete. 
Nevertheless, \probLCSS can have a polynomial \emph{Turing kernel}.

Let $f\colon\mathbb{N}\rightarrow \mathbb{N}$.  A \emph{Turing kernelization} of size $F$ for a parameterized problem is an algorithm that decides whether a given instance $(x,k)$ of the problem, where $x$ is an input and $k$ is a parameter, is a yes-instance in time polynomial in $|x|+k$, when given the access to an oracle that decides whether an instance $(x',k')$, where  $|x'|+k'\leq f(k)$, is a yes-instance in a single step. Respectively, a Turing kernel is polynomial if $f$ is a polynomial.

We show that \probLCSS has a polynomial Turing kernel if $\Pi$ is the property to be a star.

\begin{theorem}
\probSSt problem admits a polynomial Turing kernelizaiton. 
\end{theorem}

\begin{proof}
Let $S$ be a solution of \probSSt on input $(G, \omega, t, w)$.  Remember that for each $x\in V(G)$, we called $F_x$ the subgraph of $G[N_G(x) \cup N_G^2(x)]$ obtained by the deletion of all edges with both endpoints in $N_G^2(x)$. From Lemma \ref{lem:SStVC} if $S$ is a $t$-secluded star of $G$ with center $x$, then $N_G(S)$ is a vertex cover of size at most $t$ of $F_x$.  Our kernelization algorithm will first compute, for each $x \in V(G)$, the graph $F_x$. Then, it performs a Buss's kernelization on each graph $F_x$. 

For a vertex $x \in V(G)$, let $W_x$ the set of vertices of degree greater than $t$ in $F_x$. Note that every vertex cover of size at most $t$ of $F_x$ must contain $W_x$. Hence, if $|W_x|$ is greater than $t$, then $x$ can not be the center of a $t$-secluded star.

Suppose now that $|W_x|\leq t$ and let $t' = t - |W_x|$. Note that $F_x$ contains a vertex cover of size $t$ if and only if $F_x - W_x$ contains a vertex cover of size at most $t'$. Since $F_x - W_x$ is a graph of degree at most $t$, a vertex cover of size $t'$ can cover at most $t\cdot t'$ edges. In other words, if $F_x - W_x$ contains more than  $ 2t\cdot t'$ non isolated vertices, then $x$ can not be the center of a $t$-secluded star.

Suppose now that $F_x - W_x$ has at most $2t\cdot t'$ non isolated vertices. Let  $F^+_x$ be the subgraph of $G[N_G^2[x]]$, obtained from $F_x - W_x$ removing all the isolated nodes, and then adding $x$ with all its incident edges. Note that $F^+_x$ contains at most $2t^2 +1$ vertices.  Now let $I_x$ be the set of isolated vertices of $F_x - W_x$ contained in $N_G(x)$. Since $N_G(I_x) \setminus \{x\}$ is contained in $W_x$, and $W_x$ is contained in $N_G(S)$ for every $t$-secluded star $S$ with center $x$, we deduce that $I_x$ may be contained in any such star $S$. In other words, if $S$ is a star with center $x$, then $S$ is a solution of \probSSt on input $(G, \omega, t, w)$ if and only if $S - I_x$ is a solution of \probSSt on input $(F^+_x, \omega^+, t', w')$, where  $\omega^+(v) = \omega(v)$ for all $v\in V(F^+_x)$ and $w' = w - \omega(I_x)$.

Now let $F'_x$ be the graph obtained from $F^+_x$, attaching to each node $u$ in $N_G^2(x) \cap V(F^+_x)$ a clique $K_u$ of size $2t$, where all vertices of $K_u$ are adjacent to $u$. Note that $F^+_x$ is a graph of size at most $ 4t^3 +1$. Moreover, every $t'$-secluded star in $F'_x$ has center $x$.  Indeed, if the center is some vertex in $N(x)$ or $N_G^2(x)$, then $S$ intersects or has as neighbors more than $t$ vertices in some clique $K_u$. Let $\omega'$ be a function on $V(F'_x)$ such that $\omega'(v) = \omega(v)$ if $v \in V(F'_x) \cap V(G)$ and $\omega'(v) = 1$ if $v$ is in one of the cliques adjacent to a node of $N_G^2(x) \cap V(F'_x)$. 

We conclude that $(F'_x, \omega', t', w')$ is a yes-instance of \probSSt for some $x \in V(G)$ if and only if $(G, \omega, t, w)$ is a yes-instance of \probSSt.

The kernelization algorithm computes for each $x \in V(G)$ the graph $F_x$ and the set $W_x$. If $W_x \cap N_G^2[x]$ contains more than $t$ nodes the algorithm rejects $x$ and continue with another vertex of $V(G)$.  If $|W_x|\leq t$ it computes $F^+_x$ deleting all the isolated vertices of $F_x - W_x$. If $F^+_x$ contains more than $2t\cdot t_1 +1$ nodes the algorithm rejects $x$ and continue with another vertex of $V(G)$. Finally, the algorithm computes $F_x'$, $t'$, $w'$ and calls the oracle on input $(F'_x, \omega', t', w')$. If the oracle answers that $(F'_x, \omega', t', w')$ is a yes-instance, the algorithm decides that $(G, \omega, t, w)$ is a yes-instance. Otherwise the algorithm continue with another vertex of $V(G)$. The algorithm ends when some oracle accepts, or all nodes are rejected. 
\end{proof}


\end{document}